\documentclass[11pt]{article}

\usepackage{palatino}
\usepackage{pgfplots}
\usepackage{tikz, tikzpeople}
\usetikzlibrary{arrows}
\usetikzlibrary{shapes.multipart}
\usetikzlibrary{fit}
\usetikzlibrary{shapes.geometric,positioning}
\usetikzlibrary{decorations.pathmorphing}
\usetikzlibrary{patterns}

\tikzset{snake it/.style={decorate, decoration=snake}}
\tikzset{meter/.append style={draw, inner sep=10, rectangle, font=\vphantom{A}, minimum width=30, line width=.8,
 path picture={\draw[black] ([shift={(.1,.3)}]path picture bounding box.south west) to[bend left=50] ([shift={(-.1,.3)}]path picture bounding box.south east);\draw[black,-latex] ([shift={(0,.1)}]path picture bounding box.south) -- ([shift={(.3,-.1)}]path picture bounding box.north);}}}
\def\bbsmatrix#1{\begin{bsmallmatrix}#1\end{bsmallmatrix}}

\usepackage{pgfplots}
\usetikzlibrary{shapes.multipart}
\usetikzlibrary{fit}
\usetikzlibrary{shapes.geometric,positioning}
\usetikzlibrary{calc}
\usetikzlibrary{patterns.meta}
\usepackage{caption}
\usepackage{subcaption}
\usepackage{url}
\usepackage{hyperref}
\usepackage{stfloats}
\usepackage{braket}
\usepackage{booktabs}
\usepackage{siunitx}

\usepackage[utf8]{inputenc}
\usepackage{mathtools}
\usepackage{graphicx}
\usepackage{algorithm}
\usepackage{cite}
\usepackage{amsmath,amssymb,amsfonts, amsthm, mathrsfs}
\usepackage{algorithmic}
\usepackage{textcomp}
\usepackage{xcolor}
\usepackage{enumitem}

\newtheorem{lemma}{Lemma}
\newtheorem{theorem}{Theorem}

\newtheorem{corollary}{Corollary}
\newtheorem{proposition}{Proposition}

\newcommand\rk{\normalfont{\mbox{rk}}}
\newcommand{\dec}{\mbox{\footnotesize dec}}
\newcommand\TQC{\normalfont{\scalebox{0.65}{\mbox{TQC}}}}

\newcommand\AME{\normalfont{\scalebox{0.65}{\mbox{AME}}}}
\newcommand\init{\normalfont {\sf init}}

\newcommand\achi{\normalfont {\sf achi}}

\newcommand\encoder{\normalfont{\sf{Enc}}}
\newcommand\decoder{\normalfont{\sf{Dec}}}

\newcommand\conv{\normalfont{\rm{conv}}}

\DeclareMathOperator{\tr}{tr}
\DeclareMathOperator{\Tr}{Tr}
\DeclareMathOperator{\Image}{Im}

\allowdisplaybreaks

\title{Capacity of Summation over a Symmetric Quantum Erasure MAC with  Partially Replicated Inputs}
\author{Yuhang Yao, Syed A. Jafar\\
{\small Center for Pervasive Communications and Computing (CPCC)}\\
{\small University of California Irvine, Irvine, CA 92697}\\
{\small \it Email: \{yuhangy5, syed\}@uci.edu}
}

\date{}      

\begin{document}
\maketitle

\begin{abstract}
The optimal quantum communication cost of computing a classical  sum of distributed sources is studied over a quantum erasure multiple access channel (QEMAC). $K$ classical messages comprised of finite-field symbols are distributed across $S$ servers, who also share quantum entanglement in advance. Each server $s\in[S]$ manipulates its quantum subsystem $\mathcal{Q}_s$ according to its own available classical messages and  sends $\mathcal{Q}_s$ to the receiver who then computes the sum of the messages based on a joint quantum measurement. The download cost from Server $s\in [S]$ is the logarithm of the dimension of $\mathcal{Q}_s$. The rate $R$ is defined as the number of instances of the sum computed at the receiver, divided by the total download cost from all the servers. The main focus is on the symmetric setting with $K= {S \choose \alpha} $ messages where each message is replicated among a unique subset of $\alpha$ servers, and the answers from any $\beta$ servers may be erased. If no entanglement is initially available to the receiver, then we show that the capacity (maximal rate) is precisely $C= \max\left\{ \min \left\{ \frac{2(\alpha-\beta)}{S}, \frac{S-2\beta}{S} \right\}, \frac{\alpha-\beta}{S} \right\}$. The capacity with arbitrary levels of prior entanglement $(\Delta_0)$ between the $S$ data-servers and the receiver is also characterized, by including an auxiliary server (Server $0$) that has no classical data, so that the communication cost from Server $0$ is a proxy for the amount of receiver-side entanglement that is available in advance. The challenge on the converse side resides in the optimal application of the weak monotonicity property, while the achievability combines ideas from classical  network coding and treating qudits as classical dits, as well as new constructions based on the $N$-sum box abstraction that rely on absolutely maximally entangled quantum states.
\end{abstract}

{\let\thefootnote\relax\footnote{By quantum ``MAC" we simply mean a many-to-one communication/computation scenario similar to \cite{Yao_Jafar_Sum_MAC}. We do not consider noisy multiple access channels.}\addtocounter{footnote}{-1}}

\allowdisplaybreaks
\section{Introduction}
Described by Schrodinger in 1935 \cite{Schrodinger1935} as the foremost distinctive feature of quantum mechanics, entanglement -- especially among \emph{many} parties --- remains one of the theory's most intriguing aspects. From a network information theoretic perspective, a key objective is to quantify  multiparty entanglement in terms of its utility as a resource, by the gains in communication efficiency that are enabled by entanglement, not only for communication tasks but more generally for multiparty \emph{computation} tasks. 

Going back at least four decades to the 1979 work of Korner and Marton in \cite{Korner_Marton_sum}, it is understood in classical network information theory, that the generalization from communication tasks, where the receivers need to recover various desired input messages, to \emph{computation} tasks, where the receivers need to recover some \emph{functions} of the distributed input messages, is highly non-trivial. Remarkably, this is the case even when the function to be computed is simply a sum of inputs and the network is simply a multiple access channel, as in \cite{Korner_Marton_sum}. Despite the challenges, with accelerating trends towards distributed computing there has been an explosion of interest in the capacity of communication networks when used for (typically, linear) computation tasks. Examples from an abundance of literature on this topic include works on the capacity of \emph{sum-networks} \cite{Rai_Dey, Ramamoorthy_Langberg, Appuswamy1}, \emph{network function computation} \cite{Kowshik_Kumar,  NetworkFC,Guang_Yeung_Yang_Li} and \emph{over-the-air computation} \cite{OTAC}, to name a few.

By the same token, there is growing interest in the fundamental limits of classical linear computation over a \emph{quantum} multiple access channel (QMAC), spurred by a variety of applications ranging from distributed quantum sensing and metrology \cite{Lloyd, GlobalMetrology, DistributedQS,Young} and quantum simultaneous message passing protocols \cite{buhrman1998quantum, QSM, PSQM} to quantum private information retrieval (QPIR) \cite{song_multiple_server_PIR, song_colluding_PIR, QMDSTPIR, hayashi2021computation,Yao_Jafar_Sum_MAC, aytekin2023quantum}. In particular, the $\Sigma$-QMAC model introduced in \cite{Yao_Jafar_Sum_MAC} is the most closely related to our work in this paper. Here we explore its generalization to a $\Sigma$-QEMAC, i.e., a $\Sigma$-QMAC that allows \emph{erasures}.
 
Fig. \ref{fig:example} illustrates an example of a $\Sigma$-QMAC setting of \cite{Yao_Jafar_Sum_MAC}. 
\begin{figure}[t]
\begin{center}
\begin{tikzpicture}
\coordinate (O) at (0,0){};
\node [draw, rectangle, rounded corners, aspect=0.2,shape border rotate=90,fill=black!10, text=black, inner sep =0cm, minimum height=1.1cm, minimum width=1.4cm, above right = 0cm and -5cm of O, align=center] (0)  {\footnotesize Server 0\\[-0.1cm]\footnotesize $-$};
\node [draw, rectangle, rounded corners, aspect=0.2,shape border rotate=90,fill=black!10, text=black, inner sep =0cm, minimum height=1.1cm, minimum width=1.4cm, above right = 0cm and -3cm of O, align=center] (ABC)  {\footnotesize Server 1\\[-0.1cm]\footnotesize ${\sf A},{\sf B},{\sf C}$};
\node [draw, rectangle, rounded corners, aspect=0.2,shape border rotate=90,fill=black!10, text=black, inner sep =0cm, minimum height=1.1cm, minimum width=1.4cm, above right = 0cm and -1cm of O, align=center] (ADE)  {\footnotesize Server 2\\[-0.1cm]\footnotesize ${\sf A},{\sf D}, {\sf E}$};
\node [draw, rectangle, rounded corners, aspect=0.2,shape border rotate=90,fill=black!10, text=black, inner sep =0cm, minimum height=1.1cm, minimum width=1.4cm, above right = 0cm and 1cm of O, align=center] (BDF)  {\footnotesize Server 3\\[-0.05cm]\footnotesize ${\sf B}, {\sf D}, {\sf F}$};
\node [draw, rectangle, rounded corners, aspect=0.2,shape border rotate=90,fill=black!10, text=black, inner sep =0cm, minimum height=1.1cm, minimum width=1.4cm, above right = 0cm and 3cm of O, align=center] (CEF)  {\footnotesize Server 4\\[-0.05cm]\footnotesize ${\sf C}, {\sf E},{\sf F}$};

 \node[alice, draw=black, text=black, minimum size=0.8cm, inner sep=0, above right = -2.5cm and -0.5cm of O] (U) at (0,0) {};

\node [draw, ellipse, dotted, aspect=0.2,shape border rotate=90,fill=white!10, text=blue, inner sep =0cm, minimum height=0.75cm, inner sep=0.1cm, minimum width=2cm, align=center, above right = 2cm and -1.5cm of O] (E)  {\footnotesize   Quantum\\[-0.1cm] \footnotesize Entanglement};

\draw [->,  black!30,  blue, snake it, out=180, in=70] (E) to (0.north);  
\draw [->,  black!30,  blue, snake it, out=200, in=80] (E) to (ABC.north);  

\draw [->,  black!30,  blue, snake it, out=270, in=90] (E) to (ADE.north);  

\draw [->,  black!30,  blue, snake it, out=320, in=100] (E) to (BDF.north);  

\draw [->,  black!30,  blue, snake it, out=0, in=110] (E) to (CEF.north);  

\draw[->, blue, snake it] (0.south) to node[left=0.2, pos=0.3,  text=blue, align=center]{\footnotesize $\mathcal{Q}_0$}(U) ;
\draw[->, blue, snake it] (ABC.south) to node[left=0.2, pos=0.3,  text=blue, align=center]{\footnotesize $\mathcal{Q}_1$}(U) ;
\draw[->, blue, snake it] (ADE.south)to node[left=0, pos=0.3,  text=blue]{\footnotesize $\mathcal{Q}_2$}(U);
\draw[->, blue, snake it] (BDF.south)to node[right=0.1, pos=0.3,  text=blue]{\footnotesize $\mathcal{Q}_3$}(U);
\draw[->, blue, snake it] (CEF.south)to node[right=0.3, pos=0.3,  text=blue,align=center]{\footnotesize $\mathcal{Q}_4$}(U);

\node (Ans) [right=0.75cm of U]{\footnotesize $({\sf A}_\ell+{\sf  B}_\ell+{\sf C}_\ell+\cdots+{\sf F}_\ell)_{\ell\in\{1,\cdots,L\}}$};
\draw [->] (U)--(Ans); 
\end{tikzpicture}
\end{center}
\vspace*{-3mm}\caption{A $\Sigma$-QMAC example with $S=4$ data-servers, one auxiliary server (Server $0$), and $K = 6$ data streams $({\sf A}, {\sf B}, {\sf C}, \cdots, {\sf F})$. A $\Sigma$-QEMAC setting   allows some of the $\mathcal{Q}_i$ (e.g., any one of $\mathcal{Q}_1, \mathcal{Q}_2, \mathcal{Q}_3, \mathcal{Q}_4$) to be erased.}\vspace*{-3mm}\label{fig:example}
\end{figure}
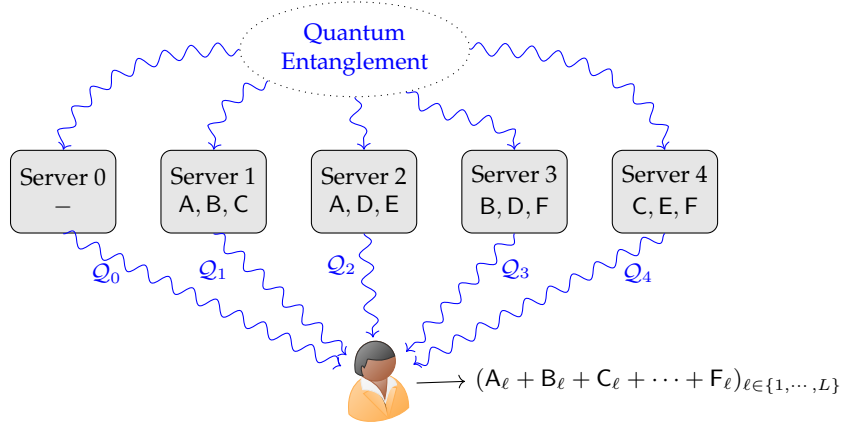
In the example shown in the figure, there are $K=6$  data streams, denoted as ${\sf A}, {\sf B}, \cdots, {\sf F}$, comprised of $\mathbb{F}_d$ symbols. Various subsets of these data streams are made available to a set of servers who also have \emph{entangled} quantum resources $(\mathcal{Q}_0,\cdots,\mathcal{Q}_4)$ distributed among them in advance (independent of the data). Any classical information available to a server may be suitably encoded into its quantum system through local operations, after which the quantum systems are sent to the receiver (Alice) through ideal (noise-less) quantum channels. Note that Server $0$, which has no data inputs and simply forwards its quantum subsystems to Alice, is an auxiliary server that may be included to  model any prior entanglement that the data-servers share in advance with Alice. The cost of communicating each quantum subsystem $\mathcal{Q}_s$ to Alice (say, in qubits) is equal to the (base $2$) logarithm of its dimension, i.e., $\log|\mathcal{Q}_s|$. Based on a joint measurement of the received quantum systems, Alice must be able to recover the sum of the classical data streams. Communication rates are measured in terms of the number of instances of the sum that are computed by Alice, per qubit of communication cost from the data-servers, for any given rate from the auxiliary server (see Section \ref{sec:symmetric} for formal definitions). The supremum of achievable rates defines the capacity.

As in the classical setting, the generalization from communication to computation tasks is also non-trivial in the quantum setting. For example, while the capacity of classical communication over a quantum multiple access channel (QMAC) is known under many fairly general noisy channel models for well over a decade  \cite{yard2008capacity, hsieh2008entanglement}, the recent capacity characterization for the  $\Sigma$-QMAC in \cite{Yao_Jafar_Sum_MAC} is limited to ideal (noise-free) channels. The aforementioned studies of QPIR and Quantum Metrology also typically choose an ideal channel model. While the classical communication capacity of an ideal (noise-free) QMAC is trivial, what makes these problems non-trivial even with ideal channels is the multiparty computation aspect --- that Alice needs to recover a function of the distributed inputs, instead of the inputs themselves. Indeed, instead of noisy channels, the focus in the $\Sigma$-QMAC is on quantifying the utility of multiparty entanglement as a resource that facilitates desired linear computations, by comparing the capacity with and without entanglement. What makes the $\Sigma$-QMAC especially interesting is the interplay between classical dependencies (replication of data across servers, mixing of inputs due to the desired computation), and multiparty quantum entanglement, which  creates opportunities for synergistic gains from a confluence of  ideas from the classical domain such as interference alignment \cite{Jafar_FnT}, network coding \cite{Ahlswede_network_information_flow}, and network function computation \cite{wei2023robust}  on one hand, and the quantum domain such as teleportation, superdense coding \cite{werner2001all}, and the $N$-sum box abstraction \cite{Allaix_N_sum_box} on the other. 
The study of the $\Sigma$-QMAC in \cite{Yao_Jafar_Sum_MAC} takes a perspective similar to degrees-of-freedom (DoF) analyses in wireless networks \cite{Jafar_FnT}, where noise is de-emphasized and the focus is entirely on quantifying the distributed accessibility of signal dimensions as the key resource. Note the similarity to the resource-theoretic accounting of quantum entanglement in facilitating classical and quantum communication in a \emph{two-party} setting, also with simplified channels, in \cite{grassl2022entropic, mamindlapally2023singleton}. The capacity of the $\Sigma$-QMAC in \cite{Yao_Jafar_Sum_MAC}, and its generalization to include erasures as in this work, can thus be seen as attempts to provide a \emph{resource-theoretic} accounting of \emph{multiparty} entanglement for an elemental multiparty computation task --- computing a  sum of distributed inputs.

Generalizing the $\Sigma$-QMAC model to a $\Sigma$-QEMAC, where erasures are possible, leads to interesting questions. For example, consider the counterpoint that multiparty entanglement in the presence of erasures could be  a double-edged sword. On one hand, the quantum systems sent by distributed servers to Alice are capable of facilitating more efficient linear computation if they are entangled than if they are unentangled. But on the other hand, because Alice in general needs to jointly measure the entangled systems, it is possible that the loss of some of the sub-systems may turn the entanglement into a \emph{disadvantage}, i.e., it is possible that the sub-systems that are received successfully may be useless without the sub-systems that are lost,  \emph{because} of their entanglement. Recall that in the point to point setting, quantum erasure correction is in general less efficient than classical erasure correction, as evident from their respective Singleton bounds \cite{mamindlapally2023singleton}. In the classical case, in order to protect $K$  symbols  against $D-1$ erasures, it suffices to encode them into $N$ symbols, such that $N \geq K+D-1$. However, in the quantum case, protecting $K$ qudits from $D-1$ erasures requires encoding them into at least $N\geq K+2(D-1)$ qudits. For example, $1$ classical symbol can be protected from $1$ erasure by simple repetition, i.e., by sending $2$ symbols, which doubles the communication cost. Protecting quantum information carried by $1$ qudit against one erasure on the other hand requires at least $3$ qudits, tripling the communication cost. Without entanglement, qudits can always be used as classical dits, allowing classical erasure codes that are more efficient. With entanglement, especially in network computation settings, we have the potential for superdense coding gain to improve communication efficiency, but can this advantage overcome the potential loss of efficiency due to the greater cost of quantum erasures? Such a question exemplifies the insights that we seek in this work. 


Going back to the example in Fig. \ref{fig:example}, and assuming for simplicity that no prior entanglement is available to Alice (i.e., ignoring Server $0$), suppose any \emph{one} of the $4$ quantum systems $(\mathcal{Q}_1,\mathcal{Q}_2,\mathcal{Q}_3,$ $\mathcal{Q}_4)$ could be erased, e.g., lost in transit or known to be corrupted and therefore discarded by Alice. Consider first a few  natural baselines for comparison. The first baseline scheme does not utilize quantum entanglement, it simply treats qudits as classical dits (TQC). The problem is reduced to a classical sum-network coding problem (cf. \cite{Rai_Dey, Ramamoorthy_Langberg}) with possible link failure as considered in \cite{wei2023robust}. Applying the result of \cite{wei2023robust}, the rate achieved is $R = 1/4$. 
As another baseline, consider a scheme based on   \cite{Yao_Jafar_Sum_MAC}. While the coding scheme in \cite{Yao_Jafar_Sum_MAC} does not allow for erasures, an erasure-tolerant scheme can be built from it as follows. In the absence of erasures,  \cite{Yao_Jafar_Sum_MAC} shows that Alice can compute ${\sf B} + {\sf C}+{\sf E}$ with only the transmission from Server $1$ and Server $2$. Alternatively, the same can also be done with only the transmission from Server $3$ and Server $4$. Implementing both alternatives together produces a scheme that allows Alice to compute ${\sf B} + {\sf C}+{\sf E}$ even if any one of the quantum subsystems $\mathcal{Q}_1,\cdots,\mathcal{Q}_4$ is unavailable at the decoder. Similarly, Alice can compute ${\sf A} +{\sf D} +{\sf F}$ using only Servers $1$ and $3$, or by using only Servers $2$ and $4$. Combining both alternatives  allows Alice to also compute ${\sf A} +{\sf D} +{\sf F}$ while tolerating the erasure of any one quantum system, i.e., by tolerating any one unresponsive server. Overall, this approach also achieves the rate $R=1/4$. Note that while this approach combines entanglement-assisted coding schemes, the redundancy required to tolerate erasures is such that the benefits of entanglement are lost, as reflected in the fact that the rate achieved is the same as without entanglement assistance in the first baseline scheme.
Fortunately, it turns out that the optimal (capacity achieving) scheme improves upon these baselines, and the benefits of entanglement do \emph{not} vanish due to erasures. Applying the main result of this work from Theorem \ref{thm:achievability} to this particular example will show that the rate $R=1/2$ is achievable and  optimal. With this brief preview of our main result, let us now summarize the relevant background that we need for this work.

\section{Preliminaries}
\subsection{Miscellaneous}
$\mathbb{N}$ denotes the set of positive integers. 
For $n_1, n_2\in \mathbb{N}$, $[n_1:n_2]$ denotes the set $\{n_1,n_1+1,\cdots,n_2\}$ if $n_1\leq n_2$ and $\emptyset$ otherwise. $[n]\triangleq [1:n]$ for $n\in \mathbb{N}$. 
$\mathbb{C}$ denotes the set of complex numbers. $\mathbb{R}_+$ denotes the set of non-negative real numbers. $\mathbb{F}_q$ denotes the finite field with $q=p^r$ a power of a prime. 
Define compact notations $A^{[n]} \triangleq (A^{(1)}, A^{(2)},\cdots, A^{(n)})$ and $A_{[n]} \triangleq (A_1, A_2, \cdots, A_n)$. $\mathcal{S}^{a\times b}$ denotes the set of $a\times b$ matrices with elements in $\mathcal{S}$.  
For an element $x\in \mathbb{F}_q = \mathbb{F}_{p^r}$, define $\mbox{tr}(x)$ as the \emph{field trace} of $x$ that maps it to an element in $\mathbb{F}_p$.
For a set $\mathcal{A}$, the set of its cardinality-$m$ sub-sets is denoted as $\binom{\mathcal{A}}{m} \triangleq \{\mathcal{A} \subseteq \mathcal{A} \mid |\mathcal{A}| = m \}$. The notation $2^{\mathcal{A}}$ denotes the power set of $\mathcal{A}$. The notation $f:\mathcal{A} \mapsto \mathcal{B}$ denotes a map $f$ from  $\mathcal{A}$ to  $\mathcal{B}$. 
$M^\dagger$ denotes the conjugate transpose of a matrix $M$, and $\Tr(M)$ denotes the trace of a square matrix $M$. 
$\mathcal{H}$ denotes a finite-dimensional Hilbert space. $\mathcal{D}(\mathcal{H})$ denotes the set of density operators acting on $\mathcal{H}$, i.e., the set of positive semi-definite matrices that have trace one. $\mathcal{L}(\mathcal{H})$ denotes the set of square linear operators acting on $\mathcal{H}$. 
${\sf Pr}(E)$ denotes the probability of an event $E$.

\subsection{Quantum information and entropy}
For a classical random variable $X$, the support set and probability mass function are denoted as $\mathcal{X}$ and $p_X(x)$ by default.
For a quantum system $A$, $\mathcal{H}_A$ is used to denote its underlying Hilbert space by default, so that a pure state of $A$ is described by a unit vector in $\mathcal{H}_A$. A general (mixed) state of $A$ is described a density operator $\rho_A \in \mathcal{D}(\mathcal{H}_A)$. We may write $\mathcal{D}(\mathcal{H}_A)$ as $\mathcal{D}(A)$ for simplicity.
A classical random variable $X$ may be regarded as a special quantum system with density operator $\rho_X = \sum_{x\in \mathcal{X}} p_X(x)\ket{x}\bra{x}_X$ where $\{\ket{x}\}_{x\in \mathcal{X}}$ is an orthonormal basis of $\mathcal{H}_X$.
For quantum systems $A_1,A_2,\cdots, A_S$, the composite system is compactly written as $A_1A_2\cdots A_S$. The underlying Hilbert space for the composite system is $\mathcal{H}_{A_1A_2\cdots A_S} = \mathcal{H}_{A_1}\otimes \mathcal{H}_{A_2} \otimes \cdots \otimes \mathcal{H}_{A_S}$. 

Let $S_v(\rho) \triangleq -\Tr(\rho \log \rho)$ be the von-Neumann entropy of a density operator $\rho$.
For a quantum system $A$ in the state $\rho_A \in \mathcal{D}(A)$, define $H(A)_{\rho_A} \triangleq S_v(\rho_A)$ as the entropy of $A$. The subscript $\rho_A$ in $H(A)_{\rho_A}$ may be omitted for compact notation when it is obvious from the context.

For a bipartite quantum system $AB$ in the (joint) state $\rho_{AB}\in \mathcal{D}(AB)$, the reduced state on $A$ is described by $\rho_A = \Tr_B(\rho_{AB}) \in \mathcal{D}(A)$, where $\Tr_B(\rho_{AB})$ is the \emph{partial trace} with respect to system $B$. $H(A)_{\rho_{AB}} \triangleq S_v(\Tr_B(\rho_{AB}))$ and similarly $H(B)_{\rho_{AB}} \triangleq S_v(\Tr_A(\rho_{AB}))$. 

Define the conditional entropy ($A$ conditioned on $B$) as $H(A\mid B)_{\rho_{AB}} \triangleq H(AB)_{\rho_{AB}} - H(B)_{\rho_{AB}}$. 
Define the mutual information (between $A$ and $B$) as $I(A;B)_{\rho_{AB}} \triangleq H(A)+ H(B) -H(AB) = H(A) - H(A\mid B) = H(B) - H(B\mid A)$. 
For a tripartite quantum system $ABC$ in the state $\rho_{ABC}$, the conditional mutual information is defined as $I(A;B\mid C)_{\rho_{ABC}} \triangleq H(A\mid C) + H(B\mid C) - H(AB\mid C) =  H(AC) + H(BC) - H(C) - H(ABC) = H(A\mid C) - H(A\mid BC) = H(B\mid C) - H(B\mid AC) = I(A;BC)-I(A;C) = I(AC;B)- I(C;B)$.

For a classical-quantum system $XA$ where $X$ is classical, if conditioned on $X=x$, $A$ is in the state $\rho_A^{(x)}$ for $x\in \mathcal{X}$, then $XA$ is in the joint state $\rho_{XA} \triangleq \sum_{x\in \mathcal{X}}p_X(x) \ket{x}\bra{x}_X \otimes \rho_A^{(x)}$. In this case, the reduced state on $A$ becomes $\rho_A = \Tr_X(\rho_{XA}) = \sum_{x\in \mathcal{X}}p_X(x) \rho_A^{(x)}$. 
We say that $H(A)_{\rho_{A}^{(x)}}$ is the entropy for $A$ conditioned on $X=x$, which is also written as $H(A\mid X=x)_{\rho_{XA}}$. 
With these definitions, $H(A\mid X)_{\rho_{XA}} = \sum_{x\in \mathcal{X}}p_X(x) H(A\mid X=x)_{\rho_{XA}}$ (e.g., see \cite[Eq. (11.54)]{Wilde_2017}), similar to classical information measures. 

We similarly define $H(A\mid B, X=x)_{\rho_{XAB}} = H(A\mid B)_{\rho_{AB}^{(x)}}$, $I(A;B\mid X=x)_{\rho_{XAB}} \triangleq I(A;B)_{\rho_{AB}^{(x)}}$ and $I(A;B\mid C, X=x)_{\rho_{XABC}} \triangleq I(A;B\mid C)_{\rho_{ABC}^{(x)}}$ as the corresponding information measures with conditioning on $X=x$. It is also true that $H(A\mid BX)_{\rho_{XAB}} = \sum_{x\in \mathcal{X}} p_X(x) H(A\mid B, X=x)_{\rho_{XAB}}$, $I(A;B\mid X)_{\rho_{XAB}} = \sum_{x\in \mathcal{X}}p_X(x) I(A;B \mid X=x)_{\rho_{XAB}}$ and $I(A;B\mid CX)_{\rho_{XABC}} = \sum_{x\in \mathcal{X}} p_X(x) I(A;B\mid C, X=x)_{\rho_{XABC}}$.

\subsection{Quantum channel}
A quantum channel is a completely positive trace-preserving (CPTP) map $\Phi\colon \mathcal{L}(\mathcal{H}) \mapsto \mathcal{L}(\mathcal{H}')$ between two spaces of square linear operators. It has a Choi-Kraus decomposition as $\Phi(M) = \sum_{i=1}^{n}V_i M V_i^\dagger$ for any $M\in \mathcal{L}(\mathcal{H})$, where $V_i$ is a linear operation that takes $\mathcal{L}(\mathcal{H})$ to $\mathcal{L}(\mathcal{H}')$ for $i\in [n]$, such that $\sum_{i=1}^n V_i^\dagger V_i = I$ is the identity matrix \cite[Thm. 4.4.1]{Wilde_2017}. $\{V_i\}$ are referred to as the Kraus operators of the quantum channel $\Phi$.
For two quantum channels $\Phi_A \colon \mathcal{L}(\mathcal{H}_A) \mapsto \mathcal{L}(\mathcal{H}_A')$ and $\Phi_B\colon \mathcal{L}(\mathcal{H}_B) \mapsto \mathcal{L}(\mathcal{H}_B')$, $\Phi_A\otimes \Phi_B \colon \mathcal{L}(\mathcal{H}_A \otimes \mathcal{H}_B) \mapsto \mathcal{L}(\mathcal{H}_A'\otimes \mathcal{H}_B')$ denotes their parallel concatenation. 
Trace and partial trace are quantum channels. The identity channel is denoted as the identity matrix ${\bf I}$ which maps any input operator to the same output.

In this paper, quantum channels are used to model encoding operations and measurements on quantum systems. When we say that an operation $\Phi$ (a quantum channel) is applied to a quantum system $A$, we mean that the system $A$ with initial state $\rho_A$ has state $\rho_A' = \Phi(\rho_A)$ after the operation.

A quantum measurement $\Psi$ is a quantum channel that maps any density operator $\rho_A$ to $\Psi(\rho_A) = \sum_{x\in \mathcal{X}} \Tr(\Lambda_A^{(x)}\rho_A) \ket{x}\bra{x}_X$ with $\{\Lambda^{(x)}\}_{x\in \mathcal{X}}$ a set of matrices that define a POVM (positive operator valued measurement) and $\{\ket{x}\}$ an  orthonormal basis \cite[Sec. 4.6.6]{Wilde_2017}. $\Psi(\rho_A)$ thus corresponds to the density operator of a classical random variable $X$ with $p_X(x) = \Tr(\Lambda^{(x)}_A \rho_A)$ for $x\in \mathcal{X}$.
When we say that a quantum system $A$ in the state $\rho_A$ is measured by $\Psi$ (a quantum measurement), we define a classical random variable $X$ such that $p_X(x) = \Tr(\Lambda^{(x)}\rho)$ for $x\in \mathcal{X}$, and we refer to $X$ as the output  of that measurement. 

 An important set of operations, referred to as the Pauli operations, are defined as follows. Let $\{\ket{a}\}_{a\in \mathbb{F}_q}$ denote the standard basis of a $q$-dimensional quantum system such that $q=p^r$ is a power of a prime. For $x\in \mathbb{F}_q$, define the ${\sf X}(x)$ operation on the system such that ${\sf X}(x)\ket{a} = \ket{a+x}$ for $a \in \mathbb{F}_q$. Define the ${\sf Z}(z)$ operation on the system such that ${\sf Z}(z)\ket{b} = \omega^{\tr(bz)}\ket{b}$ for $b\in \mathbb{F}_q$.

\subsection{Useful Lemmas}
For our converse proofs, the following lemmas will be useful. 
\begin{lemma}[No-signaling \cite{Peres_QIT, wiki:No-communication_theorem}] \label{lem:no_signal}
	Let $XAB$ be a classical-quantum system where $X$ is classical.  Say $AB$ is in the state $\rho_{AB}^{\init} \in \mathcal{D}(AB)$ initially. Let $\{\Phi_A^{(x)}\}_{x\in \mathcal{X}}$ be a set of quantum channels (to be applied to $A$). Let $\rho_{AB}^{(x)} = \Phi_A^{(x)} \otimes {\bf I}_B (\rho_{AB}^{\init})$ for $x\in \mathcal{X}$ and let $\rho_{XAB} = \sum_{x\in \mathcal{X}}p_X(x)  \ket{x}\bra{x} \otimes \rho_{AB}^{(x)}$. Then $I(B;X)_{\rho_{XAB}} = 0$.
\end{lemma}
\begin{proof}
	It suffices to show that $H(B)_{\rho_{XAB}} = H(B\mid X)_{\rho_{XAB}}$. First note that $\Tr_A(\rho_{AB}^{(x)}) = \Tr_A(\rho_{AB}^{\init})$ for $x\in \mathcal{X}$ because  only the identity operation is applied to $B$. By definition, 
	\begin{align}
		&H(B)_{\rho_{XAB}} = S_v\big(\Tr_{XA}(\rho_{XAB})\big) = S_v\big(\Tr_{A}\Tr_{X}(\rho_{XAB})\big) \notag \\
		&= S_v\big(\Tr_A(\sum_{x\in \mathcal{X}}p_X(x) \rho_{AB}^{(x)})\big) = S_v\big(\sum_{x\in \mathcal{X}} p_X(x) \Tr_A(\rho_{AB}^{(x)})\big) \notag \\
		&= S_v\big(\sum_{x\in \mathcal{X}} p_X(x) \Tr_A(\rho_{AB}^{\init})\big) = S_v\big(\Tr_A(\rho_{AB}^{\init})\big)
	\end{align}%
	On the other hand, by definition, $H(B\mid X=x)_{\rho_{XAB}} = S_v\big(\Tr_A(\rho_{AB}^{(x)})\big) = S_v\big(\Tr_{A}(\rho_{AB}^{\init})\big)$ for all $x\in \mathcal{X}$, and thus $H(B\mid X)_{XAB} = \sum_{x\in \mathcal{X}}p_X(x) H(B\mid X=x)_{\rho_{XAB}} = S_v\big(\Tr_A(\rho_{AB}^{\init})\big)$.
\end{proof}

\begin{lemma}[Holevo Bound \cite{holevo1973bounds}] \label{lem:Holevo}
	Let $XA$ be a classical-quantum system where $X$ is classical,  $\rho_A^{(x)} \in \mathcal{D}(A)$ for $x\in \mathcal{X}$ and $\rho_{XA} = \sum_{x\in \mathcal{X}}p_X(x) \ket{x}\bra{x} \otimes \rho_{A}^{(x)}$. 
	Let $\Psi$ be a quantum measurement that measures $A$ in the state $\rho_{A} = \Tr_{X}(\rho_{XA})$, and denote the output as $Y$. Then, $I(X;Y)\leq I(X;A)_{\rho_{XA}}$.
\end{lemma}
\begin{proof}
	Since the density operator for $Y$ after the measurement is $\rho_Y = \Psi(\rho_A)$ where $\Psi$ is a given channel, quantum data processing inequality \cite[Thm. 11.9.4]{Wilde_2017} implies $I(X;Y) \leq I(X;A)_{\rho_{XA}}$.
\end{proof}

\section{Problem Statement}
\subsection{$\Sigma$-QEMAC Model}
The $\Sigma$-QEMAC problem is specified by a finite field $\mathbb{F}_d$, a set of $S$ data-servers with indices $1,\cdots, S$, an auxiliary server (Server $0$) with index $0$ that may be included to explicitly model prior shared entanglement between Alice and the data-servers, $K$ data streams, $T$ erasure patterns, the data replication map $\mathcal{W}:[K]\mapsto 2^{[S]}$ that specifies the subsets of data-servers among which a data stream is replicated, and the map $\mathcal{E}:[T]\mapsto 2^{[S]}$ that specifies the subsets of servers from which the transmissions may be unavailable (erased). Specifically, for $k\in [K]$, the $k^{th}$ data stream, ${\sf W}_k$  is comprised of symbols ${\sf W}_k^{(\ell)} \in \mathbb{F}_d$ for $\ell \in \mathbb{N}$. The data stream ${\sf W}_k$ is available at Servers $s$ for all $s\in \mathcal{W}(k)\subseteq[S]$. The coding scheme must allow Alice to recover ${\sf W}_1+{\sf W}_2+\cdots+{\sf W}_K$ given that the answers from a subset of servers $\mathcal{E}(t)\subseteq [S]$ are unavailable to the receiver (Alice), for any $t\in [T]$. Note that no data stream is available to Server $0$, and the answer from Server $0$ cannot be erased, because it represents the shared entanglement that is already available in advance to Alice. Define $\mathcal{S}=\{0\}\cup[S]$ as a compact notation for the set of all $S+1$ server indices.

A quantum coding scheme is specified by a $5$-tuple 
\begin{align*}
	\big(L, \delta_{\mathcal{S}}, \rho^{\init},  \encoder_{[S]}, \decoder \big).
\end{align*}
Here, $L\in \mathbb{N}$ is  the batch size, which specifies the number of instances of the data to be encoded together, i.e., the data to be encoded is ${\sf W}_{[K]}^{[L]} = ({\sf W}_1^{[L]}, {\sf W}_2^{[L]},\cdots, {\sf W}_K^{[L]})$. 
We use $w = (w_1,w_2,\cdots, w_K) \in \mathbb{F}_d^{KL}$ to represent a realization of $({\sf W}_1^{[L]}, {\sf W}_2^{[L]},\cdots, {\sf W}_K^{[L]})$, where $w_k$ is the realization of the $k^{th}$ data stream. The superscript `${[L]}$' over the data streams may be omitted for compact notation. Given $w$, for $s\in [S]$, let $x_s$ be the part of $w$ that is available to Server $s$, i.e., $x_s \triangleq (w_k \colon s\in \mathcal{W}(k))$.

\begin{figure}[htbp]
\center
\begin{tikzpicture}
\node[rectangle, draw, fill=gray!10, rounded corners=5, minimum width = 0.8cm, minimum height=4.1cm] (Q) at (0,-1.56) {};

\node (Q1) at (0,0) [draw, circle, dotted, aspect=0.2, fill=black!20, inner sep =0cm, minimum width=0.6cm, align=center, align=center] {$\mathcal{Q}_0$};
\node (Q2) at (0,-1.35) [draw, circle, dotted, aspect=0.2, fill=black!20, inner sep =0cm, minimum width=0.6cm, align=center, align=center] {$\mathcal{Q}_1$};
\node at (0,-2) [align=center] {$\vdots$};
\node (QS) at (0,-3.1) [draw, circle, dotted, aspect=0.2, fill=black!20, inner sep =0cm, minimum width=0.6cm, align=center, align=center] {$\mathcal{Q}_{S}$};

\node at ($(Q1.west)+(-0.4,-1.6)$) [rotate=90] {\footnotesize Entangled Quantum Systems};

\node (C1) at (1.5,0) [draw, rectangle, inner sep =0.1cm] {${\bf I}$};
\node (C2) at (1.5,-1.35) [draw, rectangle, inner sep =0.1cm] {$\Phi_1^{(x_1)}$};
\node (CS) at (1.5,-3.1) [draw, rectangle, inner sep =0.1cm] {$\Phi_{S}^{(x_{S})}$};
\node at (1.5,-2) {$\vdots$};

\draw [color=black, thick] (Q1.east)--(C1);
\draw [color=black, thick] (Q2.east)--(C2);
\draw [color=black, thick] (QS.east)--(CS);

\node (D) at (5,-0.5) [draw, rectangle, minimum height = 1.75cm, minimum width = 0.98cm] {$\Psi_t$};
\node (E) at (5,-2.7) [draw, rectangle, minimum height = 1.4cm, minimum width = 0.98cm] {$\Tr$};

\draw [color=black, thick] (C1.east)--($(C1.east)+(2.8,0)$);

\draw [color=black, thick] (C2.east)--($(C2.east)+(0.5,0)$)--($(E.west)+(-1.5,0.4)$)--($(E.west)+(0,0.4)$);

\draw [color=black, thick] ($(C2.east)+(0,-0.8)$)--($(C2.east)+(0.5,-0.8)$)--($(E.west)+(-1.5,-0.4)$)--($(E.west)+(0,-0.4)$);

\draw [color=black, thick] ($(CS.east)$)--($(CS.east)+(0.5,0)$)--($(D.west)+(-1.5,-0.4)$)--($(D.west)+(0,-0.4)$);

\node at ($(D.west)+(-1.35,0.15)$) {$\vdots$};
\node at ($(D.west)+(-0.65,0)$) {\footnotesize $\mathcal{Q}_{\mathcal{S}\setminus \mathcal{E}(t)}$};

\node at ($(D.west)+(-1.2,-2.1)$) {$\vdots$};
\node at ($(D.west)+(-0.65,-2.25)$) {\footnotesize $\mathcal{Q}_{\mathcal{E}(t)}$};

\node (Y) [right=0.3cm of D]{$Y_t^{(w)}$};
\draw [color=black, thick] (D.east)--(Y.west);

\node [above=-0.1cm of C1]{\footnotesize (Server $0$)};
\node [above=-0.1cm of C2]{\footnotesize (Server $1$)};
\node [above=-0.1cm of CS]{\footnotesize (Server $S$)};
\node (R) [above=-0.1cm of D]{\footnotesize (Alice)};
\node [below=-0.05cm of E]{\footnotesize (Erasure)};

\draw[thick, dotted] (0.6,0.4)--(0.6,-3.5);
\node at (0.6,-3.9){$\rho^{\init}$};

\draw[thick, dotted] (2.4,0.4)--(2.4,-3.5);
\node at (2.6,-3.9){$\rho^{(w)}_{\mathcal{Q}_0\cdots \mathcal{Q}_{S}}$};

\node at (2.1,1.2){\footnotesize $\forall w=(w_1,\cdots, w_K) \in \mathbb{F}_d^{KL}, t\in [T]$};

\end{tikzpicture}
\caption{A quantum coding scheme for the $\Sigma$-QEMAC. The output measured at the receiver, $Y_t^{(w)}$, must be the sum $w_1+w_2+\cdots+w_K$.}
\end{figure}
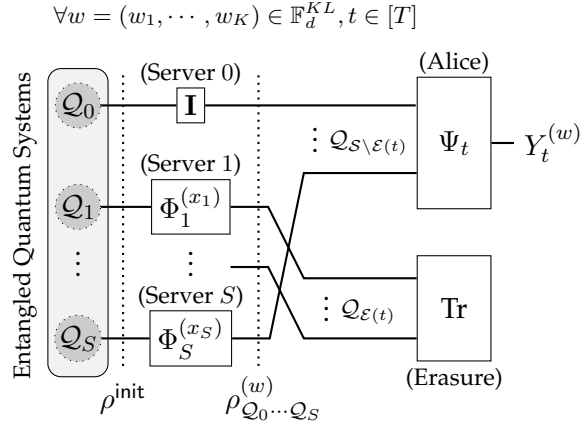

A quantum system $\mathcal{Q}_{\mathcal{S}}=\mathcal{Q}_0\mathcal{Q}_1\cdots \mathcal{Q}_S$ is prepared in advance in the initial state $\rho^{\init}$ and shared among the servers so that Server $s$ has the subsystem $\mathcal{Q}_s$ for $s\in \mathcal{S}$. For any data realization $w \in \mathbb{F}_d^{KL}$, the encoder $\encoder_s$ at Server $s\in [S]$ is represented by a quantum channel $\Phi_{s}^{(x_s)}$ that depends on $x_s$. The output dimension is upper bounded by $\delta_s, s\in\mathcal{S}$ for all $w$.  For $w\in \mathbb{F}_d^{KL}$,
\begin{align}
	 \rho_{\mathcal{Q}_0\cdots \mathcal{Q}_S}^{(w)} \triangleq {\bf I} \otimes \Phi_1^{(x_1)} \otimes \cdots \otimes \Phi_S^{(x_S)} \big( \rho^{\init} \big)
\end{align}
denotes the output state for the joint quantum system $\mathcal{Q}_0\mathcal{Q}_1\cdots\mathcal{Q}_S$, after the channel $\Phi_s^{(x_s)}$ is applied to $\mathcal{Q}_s$, for each $s\in [S]$.

For $t\in [T]$, the decoder $\decoder$ is represented by a quantum measurement,  $\Psi_t$, that depends on $t$, i.e., the erasure pattern that is encountered. For the case when the answers from Servers $s\in \mathcal{E}(t)$ are unavailable at Alice, the remaining subsystem $\mathcal{Q}_{\mathcal{S}\setminus \mathcal{E}(t)}$ in the reduced state,
\begin{align*}
	\rho_{\mathcal{Q}_{\mathcal{S}\setminus \mathcal{E}(t)}}^{(w)} = \Tr_{\mathcal{Q}_{\mathcal{E}(t)}}(\rho_{\mathcal{Q}_0\cdots \mathcal{Q}_S}^{(w)})
\end{align*}
is measured by $\Psi_t$, with the output represented by the random variable $Y^{(w)}_t$.
It is required that the decoding must be always correct, i.e.,  
\begin{align} \label{eq:def_correctness}
	{\sf Pr} \big( Y^{(w)}_t = w_1+w_2+\cdots + w_K \big) = 1
\end{align}
for all data realizations $w\in \mathbb{F}_d^{KL}$, and for \emph{any} case of erasure indexed by $t\in [T]$.

For such a coding scheme, we define $\Delta_s = \log_d \delta_s/L$ as the (normalized) download cost from Server $s$ for $s\in \mathcal{S}$, and ${\bf \Delta} \triangleq (\Delta_0, \Delta_1,\cdots, \Delta_S)$ as the cost tuple achieved by the scheme. 

Let $\mathfrak{C}_L$ be the set of coding schemes with batch size $L$.
Define $\mathfrak{D}^*$ as the closure of the cost tuples achieved by the schemes in $\mathfrak{C}_L$ as $L\to \infty$. The ultimate goal is to study $\mathfrak{D}^*$.

\subsection{Symmetric setting}\label{sec:symmetric}
To counter the combinatorial complexity of the problem, we will focus especially on a particular symmetric setting of the $\Sigma$-QEMAC. For clarity, it will be useful to identify the symmetric setting by including a superscript \emph{s} (for symmety), as in $\Sigma^s$-QEMAC.  The $\Sigma^s$-QEMAC setting is specified by a tuple $(\mathbb{F}_d, S, \alpha, \beta)$. The data symbols are from $\mathbb{F}_d$. $S,\alpha, \beta$ are integers and $S\geq \alpha>\beta \geq 0$ for the problem to be feasible. 
As in the general $\Sigma$-QEMAC, the $\Sigma^s$-QEMAC contains $S$ data-servers, indexed by $[S]=\{1,\cdots, S\}$ and an auxiliary server, Server $0$ is included if prior shared entanglement is available to Alice. There are $K = \binom{S}{\alpha}$ data streams, each replicated among a unique cardinality-$\alpha$ subset of the  data-servers.
The answers from any $\beta$-subset of the data-servers may be erased. Mathematically, the storage is specified as $\mathcal{W}:[\binom{S}{\alpha}]\mapsto \binom{[S]}{\alpha}$ and the erasure map as $\mathcal{E}:[\binom{S}{\beta}] \mapsto \binom{[S]}{\beta}$, both being bijections. 
As before, note that none of the $K$ data streams is available to Server $0$ and we assume that the answer from Server $0$ cannot be erased.


For an illustration of a $\Sigma^s$-QEMAC setting, consider Figure \ref{fig:example}, where we have $S=4$ data-servers and an auxiliary Server $0$, with $\alpha=2$ so that there are $K = \binom{4}{2}=6$ classical data streams, ${\sf W_1,W_2, W_3, W_4, W_5, W_6}$, relabeled as ${\sf A,B,C,D,E,F}$, respectively, in the figure for convenience. There is a data stream corresponding to each subset of cardinality $\alpha=2$ out of the $S=4$ data-servers, that is available precisely to those $\alpha=2$ servers. Specifically, the data stream ${\sf W_1}$ (equivalently, ${\sf A}$) is available to the data-servers with indices in the set $\mathcal{W}(1)=\{1,2\}$, ${\sf B}$ to $\mathcal{W}(2)=\{1,3\}$, ${\sf C}$ to $\mathcal{W}(3)=\{1,4\}$, ${\sf D}$ to $\mathcal{W}(4)=\{2,3\}$, ${\sf E}$ to $\mathcal{W}(5)=\{2,4\}$, ${\sf F}$ to $\mathcal{W}(6)=\{3,4\}$. Say in this example $\beta=1$, which means any one of $\mathcal{Q}_1,\mathcal{Q}_2, \mathcal{Q}_3, \mathcal{Q}_4$ may be erased. Thus, there are  $T=\binom{4}{1}=4$ erasure cases that must be tolerated, comprised of all cardinality-$1$ subsets of the $4$ data-servers, i.e., the indices of the erased servers can be any one of the sets $\{1\},\{2\},\{3\},\{4\}$. Note that Server $0$ has no data stream and $\mathcal{Q}_0$ is not subject to  erasures.

Recall that the closure of the cost tuples ${\bf \Delta} = (\Delta_0,\Delta_1,\cdots, \Delta_S)$ achieved by the coding schemes is denoted as $\mathfrak{D}^*$. We are interested in the trade-off between $\Delta_0$ and $\Delta_1+\cdots+\Delta_S$, i.e., the (normalized) cost of entanglement provided by Server $0$ (modeling entanglement previously available to Alice) and the sum (normalized) download cost from the data-servers.  $R = (\Delta_1+\cdots+\Delta_S)^{-1}$ is referred to as the data-server rate (reciprocal of the sum-download cost from the data-servers). For each $\Delta_0 \in \mathbb{R}_+$, we wish to find the capacity, which is defined as $C(\Delta_0) \triangleq \max_{(\Delta_0, \Delta_1,\cdots, \Delta_S) \in \mathfrak{D}^*} R$.
Due to the symmetry among data-servers, there is no loss of generality in the assumption that $\Delta_1=\Delta_2=\cdots=\Delta_S = \Delta$.
Therefore, the optimal tradeoff is equivalently represented as $\Delta^*(\Delta_0) \triangleq \min_{(\Delta_0, \Delta, \cdots, \Delta) \in \mathfrak{D}^*} \Delta$. Note that $C(\Delta_0) = (S\Delta^*(\Delta_0))^{-1}$.

\section{Result}
\subsection{Achievability for an arbitrary $\Sigma$-QEMAC}
The following theorem states our achievability result for an arbitrary (not necessarily symmetric) $\Sigma$-QEMAC.
\begin{theorem}[General Achievability] \label{thm:achievability}
For the $\Sigma$-QEMAC with parameters $(\mathbb{F}_d, \mathcal{S},K,T, \mathcal{W}, \mathcal{E})$, we have 
$\mathfrak{D}_{\achi} \subseteq \mathfrak{D}^*$, where  
\begin{align} 
&\mathfrak{D}_{\achi}  = \conv (\mathfrak{D}_{\AME}^+ \cup \mathfrak{D}_{\TQC}^+), \\
&\mathfrak{D}^+_{x} \triangleq \left\{ (\Delta_0^+,\cdots, \Delta_S^+) \in \mathbb{R}_+^{S+1} \left| \begin{array}{l}  \exists (\Delta_0,\cdots, \Delta_S) \in \mathfrak{D}_x, ~\Delta_s^+ \geq \Delta_s, \forall s\in \mathcal{S}\end{array} \right. \hspace{-0.25cm} \right\},  x\in \{{\rm AME}, {\rm TQC} \},\\
&\mathfrak{D}_{\AME} \triangleq \left\{ 
	{\bf \Delta}\in \mathbb{R}_+^{S+1}
    \left|
    \begin{array}{l}
     \min\left\{ \sum_{s\in \mathcal{S}}\Delta_s, \sum_{s\in \mathcal{W}(k)}2\Delta_s \right\} 
      -  \sum_{s\in \mathcal{E}(t)}2\Delta_s 
     \geq 1,   \forall k\in [K], t\in [T]
    \end{array}
    \right. \hspace{-0.25cm} \right\},\label{eq:region_AME} \\
&\mathfrak{D}_{\TQC} \triangleq \left\{ 
	{\bf \Delta}\in \mathbb{R}_+^{S+1}
    \left|
    \begin{array}{l}
     \sum_{s\in \mathcal{W}(k)} \Delta_s - \sum_{s\in \mathcal{E}(t)} \Delta_s \geq 1,  \forall k\in [K], t\in [T]
    \end{array}
    \right. \hspace{-0.25cm} \right\}.\label{eq:region_TQC}
\end{align}
\end{theorem}

\noindent Theorem \ref{thm:achievability} is an inner bound (based on an achievability argument) on the optimal  tradeoff region $\mathfrak{D}^*$. The proof appears in Section \ref{proof:achievability}.  We first prove $\mathfrak{D}_{\AME} \subseteq \mathfrak{D}^*$ by constructing coding schemes that make use of quantum entanglement. The name ``AME" comes from the fact that the quantum state in the scheme is an ``absolutely maximally entangled" state (e.g., see \cite{helwig2012absolute,  huber2013structure,goyeneche2015absolutely}). The design is facilitated by the $N$-sum box protocol in \cite{Allaix_N_sum_box}. On the other hand, $\mathfrak{D}_{\TQC} \subseteq \mathfrak{D}^*$ is directly implied by a  classical network coding result in \cite{wei2023robust} together with the idea of `treating qudits as classical dits' (TQC). It then follows that $\mathfrak{D}_{\AME}^+ \subseteq \mathfrak{D}^*$ and $\mathfrak{D}_{\TQC}^+ \subseteq \mathfrak{D}^*$. Finally, by a time-sharing argument, any convex combination of the tuples in $\mathfrak{D}_{\AME}^+ \cup \mathfrak{D}_{\TQC}^+$ is also in $\mathfrak{D}^*$. This means that $\mathfrak{D}_{\achi} = \conv(\mathfrak{D}_{\AME}^+ \cup \mathfrak{D}_{\TQC}^+) \subseteq \mathfrak{D}^*$.

\subsection{Known converse bounds}
Let us recall a useful bound from \cite{mamindlapally2023singleton, grassl2022entropic}, which we present here under the framework of $\Sigma^s$-QEMAC to make the connection more transparent.

\begin{theorem}[EACQ singleton bound \cite{mamindlapally2023singleton}]  \label{thm:eacq_bound}
	Consider the $\Sigma^s$-QEMAC with $S$ data-servers and an auxiliary Server $0$, with $\alpha=S$ so that we have only $K=1$ data stream ${\sf W}$, which is available only to the data-servers, and the answers from any $\beta$ of the data-servers may be erased. Then any $(\Delta_0, \Delta_1,\cdots, \Delta_S) \in \mathfrak{D}^*$ must satisfy the bounds,
	\begin{align}
		&~~~~~~\sum_{s\in \mathcal{I}}\Delta_{s} \geq 1/2, ~~\forall \mathcal{I} \in \binom{[S]}{S-\beta}, \label{eq:eacq_1} \\
		&\mbox{and} ~~\Delta_0+\Delta_1+\cdots+\Delta_S \geq \frac{S}{S-\beta}.\label{eq:eacq_2}
	\end{align}
\end{theorem}
\noindent Note that in this theorem since there is only one data stream, computing the sum is equivalent to recovering the data stream. So the $\Sigma$-QEMAC problem reduces to a standard communication problem. Thus, this theorem essentially follows from the EACQ singleton bound \cite{mamindlapally2023singleton}, for the special case of trading \emph{qudits and entanglement} for \emph{classical information}. Since the erasure in our model is assumed to be server-wise instead of channel-wise as assumed in \cite{mamindlapally2023singleton}, for the sake of completeness we provide in this paper the proof for our model in Appendix \ref{proof:eacq_bound}. 

Although Theorem \ref{thm:eacq_bound} is not directly applicable to $\Sigma$-QEMAC with $K>1$ data streams, it can be used to find converse bounds for $K>1$ by incorporating  the cut-set argument (e.g., \cite{Appuswamy1}), i.e., separating the parties into two groups and allowing full cooperation within each group, thus reducing it to a problem with $K=1$. Since cooperation cannot hurt the converse argument, the converse for $K=1$ is a valid converse for the original problem with $K>1$. In some cases, following such a cut-set argument, the converse for the resulting point to point communication problem yields useful outer bounds on $\mathfrak{D}^*$ for the computation problem, (e.g., \cite{Yao_Jafar_Sum_MAC}). However, we will see that the cut-set based approach does not suffice in general, and new converse bounds are needed.

\subsection{The Capacity of the $\Sigma^s$-QEMAC}
Next we specialize from arbitrary $\Sigma$-QEMAC settings to symmetric settings, and present a sharp capacity characterization for the $\Sigma^s$-QEMAC, which is the main contribution of this work. Recall that the feasibility for the symmetric case requires that $S\geq \alpha > \beta \geq 0$.

\begin{theorem} \label{thm:sym}
	For $\Delta_0\in \mathbb{R}_+$, the capacity of the $\Sigma^s$-QEMAC is
		$C(\Delta_0) = (S\Delta^*)^{-1}$, 	where
	\begin{align} \label{eq:sym_optimal_cost}
		\Delta^*  &\triangleq \min_{(\Delta_0, \Delta, \cdots, \Delta)\in \mathfrak{D}^*} \Delta \notag \\
		&= \begin{cases}
			\max\left\{\frac{1}{2(\alpha-\beta)}, \frac{1-\Delta_0}{S-2\beta} \right\}, & S \geq \alpha+\beta \\
			\max\left\{\frac{1}{2(\alpha-\beta)}, \frac{1}{\alpha-\beta}-\frac{\Delta_0}{2\alpha-S}  \right\}, & S < \alpha + \beta
		\end{cases}.
	\end{align}
\end{theorem}
\noindent The proof appears in Section \ref{proof:sym}. Note that Theorem \ref{thm:sym} is a capacity result, and as such requires both a proof of achievability and a tight converse. The achievability is obtained by letting $\Delta_0\in \mathbb{R}_+$ and $\Delta_1=\Delta_2=\cdots=\Delta_S=\Delta$  in Theorem \ref{thm:achievability} and evaluating the smallest $\Delta$ with respect to $\Delta_0$ such that $(\Delta_0, \Delta,\cdots, \Delta) \in \mathfrak{D}_{\achi}$.  The smallest $\Delta$ turns out to be $\Delta^*$ in \eqref{eq:sym_optimal_cost}.  Then $C(\Delta_0) \geq (S\Delta^*)^{-1}$ establishes the inner (lower) bound for $C(\Delta_0)$. The outer (upper) bound $C(\Delta_0) \leq (S\Delta^*)^{-1}$ utilizes not only the EACQ singleton bound (Theorem \ref{thm:eacq_bound}) combined with a cut-set argument, but also a new bound that is derived in this work, where weak monotonicity \cite{linden2005new, prabhu2013exclusion}  of quantum entropy plays an important role. Insufficiency of the EACQ singleton bound and cut-set argument is demonstrated by an example in Section \ref{sec:example}.

The next two corollaries follow directly from Theorem \ref{thm:sym}. The first corollary considers the case $\Delta_0=0$, i.e., when no entanglement is shared initially between Alice and the data-servers. Essentially Server $0$ does not exist in this case. However, entanglement is still allowed among the data-servers.

\begin{corollary}[$\Delta_0=0$] \label{cor:no_helper}
	For $\Delta_0=0$, 
	\begin{align}
		C(0) &= \max\left\{ \underbrace{\min \left\{ \frac{2(\alpha-\beta)}{S}, \frac{S-2\beta}{S} \right\}}_{R_{\AME}}, \underbrace{\frac{\alpha-\beta}{S}}_{R_{\TQC}} \right\} \label{eq:achi_symmetric} \\
		&= \begin{cases} 
			\frac{2(\alpha-\beta)}{S} = R_{\AME}, & S\geq 2\alpha \\
			\frac{S-2\beta}{S}= R_{\AME}, & \alpha +\beta \leq S \leq 2\alpha \\
			\frac{\alpha-\beta}{S}= R_{\TQC}, & S \leq \alpha + \beta
		\end{cases}. \label{eq:achi_symmetric_2}
	\end{align}
\end{corollary}
$R_{\AME}$ that appears in \eqref{eq:achi_symmetric} is defined as the maximal value of $(\Delta_1+\cdots+\Delta_S)^{-1}$ (given $\Delta_0=0)$ for ${\bf \Delta} \in \mathfrak{D}_{\AME}$. Thus, $R_{\AME}$ is achieved by our quantum coding scheme. Similarly, $R_{\TQC}$ is defined as the maximal value of $(\Delta_1+\cdots+\Delta_S)^{-1}$ (given $\Delta_0=0)$ for ${\bf \Delta} = (\Delta_0, \Delta_1,\cdots, \Delta_S) \in \mathfrak{D}_{\TQC}$. Thus, $R_{\TQC}$ is achieved by treating qudits as classical dits. It can be seen from \eqref{eq:achi_symmetric_2} that when $S > \alpha+\beta$, we have $R_{\AME}>R_{\TQC}$, i.e., the optimal scheme is our proposed scheme that is facilitated by the $N$-sum box abstraction and utilizes quantum entanglement, outperforming the classical scheme which does not require quantum entanglement.

Fig. \ref{fig:numerical} illustrates the capacity (and the rates $R_{\AME}, R_{\TQC}$) from Corollary \ref{cor:no_helper} with $S=8$ data-servers and erasure levels $\beta \in \{1,2\}$ for various data replication levels $\alpha$.

\begin{figure}[ht]
\center
\begin{tikzpicture}
\begin{axis}[
      	width=0.7\textwidth,
      	height=0.5\textwidth,
      	xmin = 2, xmax = 8,
		ymin = 0, ymax = 1,
		xlabel={$\alpha$},
        ylabel={\footnotesize Capacity (Rate)},
		legend pos=north west,
    	ymajorgrids=true,
    	grid=both,
	    xtick={2,3,4,5,6,7,8},
		legend style={at={(0.64,0.42)},anchor=north west},
    	]
	
	      \addplot[color = blue!25, line width=2mm, opacity=0.4] coordinates {
      	(2,1/4)
      	(3,2/4)
      	(4,3/4)
      	(5,3/4)
      	(6,3/4)
      	(7,3/4)
      	(8,7/8)
              };
      
      \addplot[color = red!25, line width=2mm, opacity=0.4] coordinates {
      	(3,1/4)
      	(4,2/4)
      	(5,2/4)
      	(6,2/4)
      	(7,5/8)
      	(8,6/8)
              };

      \addplot[color=cyan, thick, mark=o, dashdotted, mark options={solid}] coordinates {
      	(2,1/4)
      	(3,2/4)
      	(4,3/4)
      	(5,3/4)
      	(6,3/4)
      	(7,3/4)
      	(8,3/4)
      };
      
      \addplot[color=magenta, thick, mark=o, dashdotted, mark options={solid}] coordinates {
      	(2,1/8)
      	(3,2/8)
      	(4,3/8)
      	(5,4/8)
      	(6,5/8)
      	(7,6/8)
      	(8,7/8)
      };
      
      \addplot[color=orange, thick, mark=asterisk, dashdotted, mark options={solid}] coordinates {
      	(3,1/4)
      	(4,2/4)
      	(5,2/4)
      	(6,2/4)
      	(7,2/4)
      	(8,2/4)
      };
      
      \addplot[color=violet, thick, mark=asterisk, dashdotted, mark options={solid}] coordinates {
      	(3,1/8)
      	(4,2/8)
      	(5,3/8)
      	(6,4/8)
      	(7,5/8)
      	(8,6/8)
      };

       \addlegendentry{\footnotesize Capacity, $\beta=1$}
       \addlegendentry{\footnotesize Capacity, $\beta=2$}
       \addlegendentry{\footnotesize $(R_{\AME}), \beta=1$}
       \addlegendentry{\footnotesize $(R_{\TQC}), \beta=1$}
       \addlegendentry{\footnotesize $(R_{\AME}), \beta=2$}
       \addlegendentry{\footnotesize $(R_{\TQC}), \beta=2$}
    \end{axis}

\end{tikzpicture}
   \caption{The capacity (and rates $R_{\rm AME}$ and $R_{\rm TQC}$) are shown for $S=8$, $\beta \in \{1,2\}$ versus $\alpha$ when the receiver Alice has no prior shared entanglement with the data-servers $(\Delta_0=0)$. The data-servers are still allowed to be entangled. Note that $R_{\rm TQC}$ is the capacity of classical codes.}
   \label{fig:numerical}
\end{figure}
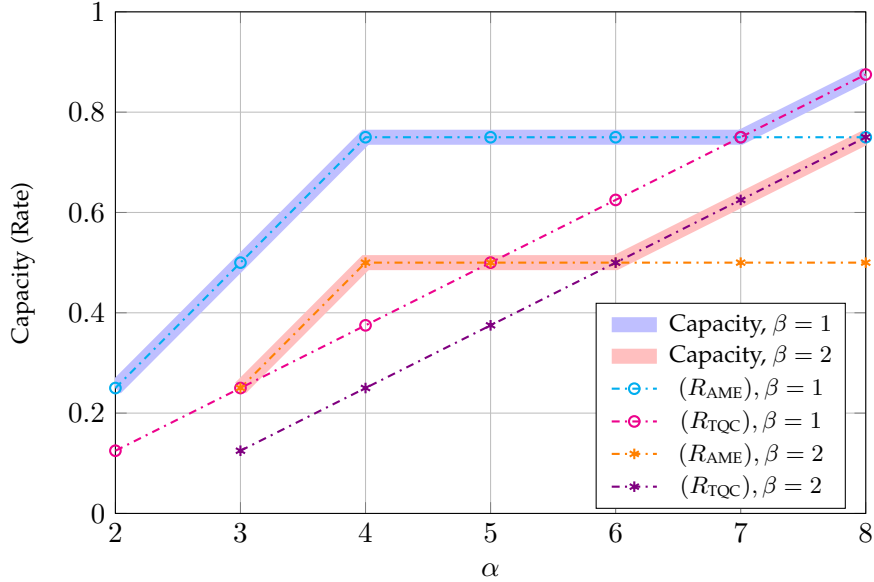

As $\Delta_0$ is increased, i.e., as the amount of initially shared entanglement between Alice and the data-servers is increased, there appears a critical threshold beyond which the capacity saturates, i.e., additional entanglement does not improve the capacity of the $\Sigma^s$-QEMAC. This threshold is highlighted in the second corollary.

\begin{corollary}[Saturation] \label{cor:sym_largeE}
	$C(\Delta_0) = \frac{2(\alpha-\beta)}{S}$ if $\Delta_0 \geq \frac{2\alpha-S}{2(\alpha-\beta)}$.
\end{corollary}
Evidently, when $\Delta_0$ is above the threshold\footnote{If $\frac{2\alpha-S}{2(\alpha-\beta)}\leq 0$, then any $\Delta_0 \geq 0$ suffices.} value $\frac{2\alpha-S}{2(\alpha-\beta)}$, then the capacity is equal to $\frac{2(\alpha-\beta)}{S}$ which is $2$ times the classical capacity $R_{\TQC}$ that appears in  \eqref{eq:achi_symmetric_2}. The factor of $2$ represents a superdense coding gain \cite{Superdense, Yao_Jafar_Sum_MAC}, and one might wonder if it can be achieved with only pairwise entanglements of each data-server with Server $0$ based on the original superdense coding protocol \cite{Superdense}. Indeed, such a pairwise superdense coding strategy can achieve the rate $R=\frac{2(\alpha-\beta)}{S}$ but it requires $\Delta_0 \geq \frac{S}{2(\alpha-\beta)}$, i.e., more entanglement than the threshold in Corollary \ref{cor:sym_largeE}. This is explained as follows. The classical scheme achieves the cost tuple $(\Delta_1,\cdots, \Delta_S) = (\frac{1}{\alpha-\beta},\cdots, \frac{1}{\alpha-\beta})$. Then using pairwise superdense coding, $\Delta_0 = \frac{S}{2(\alpha-\beta)}$ and $\Delta_1=\cdots= \Delta_S=\frac{1}{2(\alpha-\beta)}$ is  achievable. 

Let us reiterate that the saturation threshold in Corollary \ref{cor:sym_largeE} is \emph{non-trivial} because this threshold is reached with strictly smaller $\Delta_0$. Note that $S>\alpha$ when we have more than one data-stream, so $2\alpha-S < S$. Thus, pairwise superdense coding is not sufficient to achieve  the capacity in Corollary \ref{cor:sym_largeE}. Indeed, our achievable scheme utilizes entanglement across all servers, instead of merely pairwise entanglements.

Fig. \ref{fig:S4a3b2} illustrates the functional form of $\Delta^*(\Delta_0)$ defined in \eqref{eq:sym_optimal_cost} evaluated for $(S,\alpha,\beta) = (4,3,2)$.
\begin{figure}[h]
\center
\begin{tikzpicture}[xscale=2.5, yscale=2.5]
\draw [->, thick] (0,0)--(3.3,0) node [above] {$\Delta_0$};
\draw [->, thick] (0,0)--(0,1.5) node [right=0.1cm] {$\Delta$};
\node (P1) at (0,1){$\bullet$};
\node [above right=-0.3cm and -0.2cm of P1]{\footnotesize $P_1:(0,1)$};
\node (P2) at (1,1/2){$\bullet$};
\node [above  right =-0.2cm and -0.2cm of P2]{\footnotesize $P_2:\left(1,\frac{1}{2}\right)$};
\node at (3,1/2){ };
\node (P3) at (2,1/2){$\bullet$};
\node [above  right =-0.2cm and -0.2cm of P3]{\footnotesize $P_3:\left(2,\frac{1}{2}\right)$};

\draw[thick] (0,1) -- (1,1/2) -- (3.05,1/2);

\draw[pattern={north east lines}] (0,0) -- (0,1) -- (1,1/2) -- (3, 1/2) -- (3,0) -- cycle;

\draw[thick, color=white] (3,0.49) -- (3,0.01);

\node [fill=white] at (1.25, 0.25) {$(\Delta_0,\Delta,\cdots,\Delta) \notin \mathfrak{D}^*$};

\node at (1.75, 1) {$(\Delta_0,\Delta,\cdots,\Delta) \in \mathfrak{D}^*$};

\node at (3.15, 0.5) {$\Delta^*$};

\end{tikzpicture}
  \caption{The functional form of $\Delta^*(\Delta_0)$ defined in \eqref{eq:sym_optimal_cost} is illustrated for $(S,\alpha,\beta) = (4,3,2)$. $P_1$ is achieved by treating qudits as classical dits (TQC). (This follows from Corollary \ref{cor:no_helper} since for this example $S\leq\alpha+\beta$. If $S>\alpha+\beta$, entanglement is still required to achieve $\Delta^*$ for $\Delta_0=0$). $P_2$ is achieved by our quantum coding scheme (AME). $P_3$ is achievable with the original superdense coding scheme based on pairwise entanglements.} \label{fig:S4a3b2}
\end{figure}
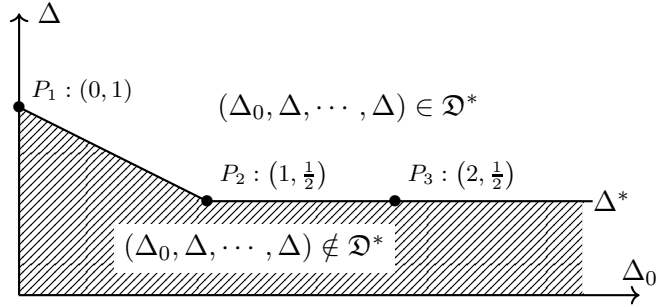

\section{Example} \label{sec:example}
\subsection{Converse bounds for symmetric $(S,\alpha,\beta) = (3,2,1)$}
Let us show the insufficiency of the cut-set argument through an example of a  $\Sigma^s$-QEMAC  with $(S,\alpha,\beta) = (3,2,1)$. There are $K= \binom{3}{2} = 3$ data streams, denoted as ${\sf A}, {\sf B}, {\sf C}$. Without loss of generality, say Server $1$ has $({\sf A}, {\sf B})$, Server $2$ has $({\sf A}, {\sf C})$ and Server $3$ has $({\sf B}, {\sf C})$. To simplify the example further, let us set $\Delta_0=0$, so that Server $0$ can be ignored. Say, $\beta = 1$, meaning that one of $\mathcal{Q}_1,\mathcal{Q}_2$ or $\mathcal{Q}_3$ may get erased.

Let us first apply the EACQ singleton bound with a cut-set argument to establish a baseline. For the cut-set argument, we need to separate the parties $\{$Alice, Server $1$, Server $2$, Server $3$$\}$ into $2$ groups. The group that contains Alice will jointly act as a receiver while the other group jointly acts as a transmitter in the resulting communication problem. Consider the following cuts.

\begin{enumerate}
	\item ``${\sf A}$  cut": We collect the parties that have the data-stream ${\sf A}$, i.e., $\{$Server $1$, Server $2$$\}$ into one group (transmitter) and $\{$Server $3$, Alice$\}$ into the other group (receiver). 
		Suppose ${\sf B} = {\sf C} = 0$ (or any constant, that is known by the receiver). Then the receiver must be able to determine ${\sf A}$, even if one of $\mathcal{Q}_1$ or $\mathcal{Q}_2$ gets erased. According to Theorem \ref{thm:eacq_bound}, with $\Delta_1=\Delta_2=\Delta_3=\Delta$, we have the bound,
			\begin{align}
				2\Delta  \geq 1/2,~~ \mbox{and}~~ 3\Delta \geq 2 \implies \Delta \geq 2/3.
			\end{align}
	\item ``${\sf A}{\sf B}$ cut": We collect the parties that have at least one of the data-streams ${\sf A},{\sf B}$, i.e., $\{$Server $1$, Server $2$, Server $3$$\}$ into the transmitter group,  leaving only $\{$Alice$\}$ in the receiver group. Suppose ${\sf C} = 0$. Then the receiver must be able to determine ${\sf A}+{\sf B}$ (which can be regarded as a single data stream). According to Theorem \ref{thm:eacq_bound}, we have
		\begin{align}
			3\Delta \geq 1/2,~~ \mbox{and}~~  3\Delta \geq 3/2 \implies \Delta \geq 1/2.
		\end{align}
	\item ``${\sf A}{\sf B}{\sf C}$ cut": The grouping can be also based on who has either (${\sf A}$ or ${\sf B}$ or ${\sf C}$), but this will yield the same partitioning and thus the same bound as the second case. 
\end{enumerate}
\begin{figure}
\center
\begin{tikzpicture}
\begin{scope}
  \node[rectangle, draw, fill=white, minimum height = 0.6cm, minimum width =0.7cm ] (S1) at (0,0) {\footnotesize Server $1$};
  \node[rectangle, draw, fill=white, minimum height = 0.6cm, minimum width =0.7cm ] (S2) at (2,0) {\footnotesize Server $2$};
  \node[rectangle, draw, fill=white, minimum height = 0.6cm, minimum width =0.7cm ] (S3) at (4,0) {\footnotesize Server $3$};
  \node[above = 1cm of S1] (A) {${\sf A}$}; 
  \node[above = 1cm of S2] (B) {${\sf B}$}; 
  \node[above = 1cm of S3] (C) {${\sf C}$};
  \node[below=1cm of S2] (sum) {${\sf A}+{\sf B}+{\sf C}$};
  \draw [thick, -latex] (A.south) -- (S1.north);
  \draw [thick, -latex] (A.south) -- (S2.north);
  \draw [thick, -latex] (B.south) -- (S1.north);
  \draw [thick, -latex] (B.south) -- (S3.north);
  \draw [thick, -latex] (C.south) -- (S2.north);
  \draw [thick, -latex] (C.south) -- (S3.north);
  \draw [thick, -latex] (S1.south) -- (sum.north) node[pos=0.5, left] {\small $\mathcal{Q}_1$};
  \draw [thick, -latex] (S2.south) -- (sum.north) node[pos=0.5, left=-0.1cm] {\small $\mathcal{Q}_2$};
  \draw [thick, -latex] (S3.south) -- (sum.north) node[pos=0.5, right=0.1cm] {\small $\mathcal{Q}_3$};
  \node[above = 2cm of S2] {\footnotesize Original sum-computation problem};
\end{scope}
  
\begin{scope}[shift={(8,0)}]
  \node[rectangle, draw, fill=white, minimum height = 0.6cm, minimum width =0.7cm ] (S1) at (0,0) {\footnotesize Server $1$};
  \node[rectangle, draw, fill=white, minimum height = 0.6cm, minimum width =0.7cm ] (S2) at (2,0) {\footnotesize Server $2$};
  \node[rectangle, draw, fill=white, minimum height = 0.6cm, minimum width =0.7cm ] (S3) at (4,0) {\footnotesize Server $3$};
  \node[above right = 1cm and 0cm of S1] (A) {${\sf A}$}; 
  \node[below=1cm of S2] (sum) {${\sf A}$};
  \draw [thick, -latex] (A.south) -- (S1.north);
  \draw [thick, -latex] (A.south) -- (S2.north);
  \draw [thick, -latex] (S1.south) -- (sum.north) node[pos=0.5, left] {\small $\mathcal{Q}_1$};
  \draw [thick, -latex] (S2.south) -- (sum.north) node[pos=0.5, left=-0.1cm] {\small $\mathcal{Q}_2$};
  \draw [thick, -latex] (S3.south) -- (sum.north) node[pos=0.5, right=0.1cm] {\small $\mathcal{Q}_3$};
  \node[above = 1.8cm of S2, align=center] {\footnotesize Reduced communication problem  \\ \footnotesize from the ``${\sf A}$" cut};
\end{scope}
\node[align=center] at (6,-1) {\footnotesize (One of $\mathcal{Q}_1, \mathcal{Q}_2$ or $\mathcal{Q}_3$ \\ \footnotesize may be erased)};

\end{tikzpicture}
\caption{Original sum-computation problem and the reduced communication problem from the ``${\sf A}$" cut.}
\label{fig:cutset}
\end{figure}
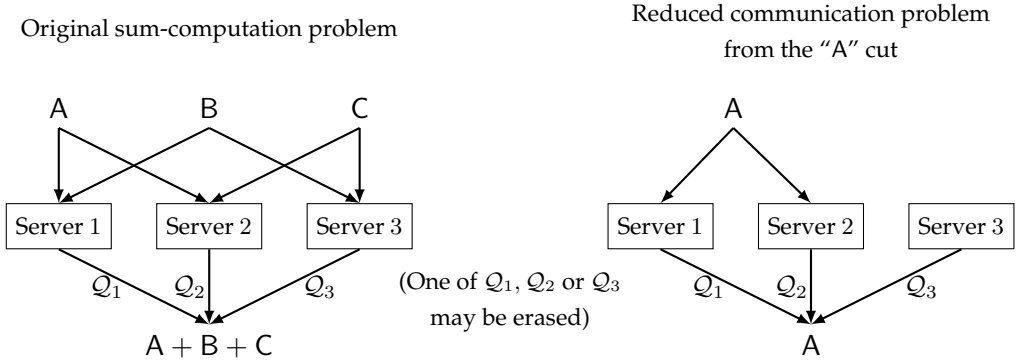

Other cuts also do not produce new bounds due to the symmetry of the problem, e.g.,  ``${\sf B}$ cut" produces the same bound as the ``${\sf A}$ cut", ``${\sf BC}$ cut" produces the same bound as the ``${\sf AB}$ cut", etc. Therefore, the cut-set and the EACQ singleton bound give us at best the converse bound $\Delta\geq 2/3$, which comes from the ``${\sf A}$ cut".  

However, this bound $\Delta \geq 2/3$ is not tight for the $\Sigma^s$-QEMAC  problem where Alice must recover the sum ${\sf A}+{\sf B}+{\sf C}$. As shown by Theorem \ref{thm:sym}, the smallest $\Delta$ given $\Delta_0=0$ for $(S,\alpha,\beta) = (3,2,1)$ is equal to $\Delta^*=1$. 

One may wonder if this is because the EACQ bound implicitly assumes that $\mathcal{Q}_3$ cannot be erased. Taking into consideration that $\mathcal{Q}_3$ may also be erased, we prove in Appendix \ref{proof:communication} that $\Delta = 3/4 < 1$ is still achievable for the communication problem reduced from the ``${\sf A}$ cut", based on superdense coding. This shows that such a cut-set argument \emph{cannot} provide stronger bounds than $\Delta \geq 3/4$, establishing the insufficiency of the cut-set argument for the $\Sigma^s$-QEMAC.
This observation shows an interesting distinction between the $\Sigma$-QEMAC, where erasures are allowed, and the original $\Sigma$-QMAC setting of \cite{Yao_Jafar_Sum_MAC}, in which no erasure is considered. A cut-set argument combined with a tight communication bound  suffices for the converse bounds for the $\Sigma$-QMAC in \cite{Yao_Jafar_Sum_MAC}, but not for the $\Sigma$-QEMAC considered in this work.

In fact the property of \emph{weak monotonicity}\footnote{Weak monotonicity is  regarded as equivalent to the strong subadditivity property of quantum entropies \cite{araki1970entropy}.} of quantum entropy (see e.g., \cite{pippenger2003inequalities, linden2005new}) plays an important role in proving a tight converse for the $\Sigma^s$-QEMAC. Intuitively, this is reminiscent of a result in \cite{prabhu2013exclusion} referred to as the exclusion principle in dense coding. However, since the $\Sigma^s$-QEMAC has distributed data streams and the information desired by Alice is a function (sum) of the data-streams rather the data-streams themselves, the connection to \cite{prabhu2013exclusion} is not straightforward.

To conclude this example, let us prove the bound $\Delta \geq 1$. Since a coding scheme must be correct for every realization of $({\sf A}, {\sf B}, {\sf C})$, it must be correct if ${\sf A}, {\sf B}, {\sf C}$ are independent random variables uniformly distributed in $\mathbb{F}_d^{L}$. Depending on the case of erasure, the measurement result is denoted as $Y_1$ (if $\mathcal{Q}_1$ gets erased); $Y_2$ (if $\mathcal{Q}_2$ gets erased) and $Y_3$ (if $\mathcal{Q}_3$ gets erased). Note that given any coding scheme, after the servers apply their encoding operations, the state of ${\sf A}{\sf B}{\sf C}\mathcal{Q}_1\mathcal{Q}_2\mathcal{Q}_3$ is determined, denoted as $\rho$. Let us consider the erasure of $\mathcal{Q}_1$ and conditioning on $({\sf B}, {\sf C})$. In this case we have,
\begin{align}
	L &= I({\sf A}, {\sf B}, {\sf C}; {\sf A} + {\sf B} + {\sf C} \mid {\sf B}, {\sf C}) \\
	&= I({\sf A}, {\sf B}, {\sf C}; Y_1 \mid {\sf B}, {\sf C}) \\
	&= I({\sf A}; Y_1 \mid {\sf B}, {\sf C}) \\
	&\leq I({\sf A}; \mathcal{Q}_2,\mathcal{Q}_3 \mid {\sf B}, {\sf C})_{\rho} \label{eq:ex_conv_1} \\
	&= \underbrace{I({\sf A}; \mathcal{Q}_3 \mid {\sf B}, {\sf C})_{\rho}}_{=0} + I({\sf A}; \mathcal{Q}_2 \mid \mathcal{Q}_3, {\sf B}, {\sf C})_{\rho}  \\
	&= I({\sf A}; \mathcal{Q}_2 \mid \mathcal{Q}_3, {\sf B}, {\sf C})_{\rho} \label{eq:ex_conv_2} \\
	&=H(\mathcal{Q}_2 \mid \mathcal{Q}_3, {\sf B}, {\sf C})_{\rho} - H(\mathcal{Q}_2 \mid \mathcal{Q}_3, {\sf A}, {\sf B}, {\sf C})_{\rho} \\
	&\leq H(\mathcal{Q}_2)_{\rho} - H(\mathcal{Q}_2 \mid \mathcal{Q}_3, {\sf A}, {\sf B}, {\sf C})_{\rho} \label{eq:ex_conv_3}
\end{align}
Step \eqref{eq:ex_conv_1} is by Lemma \ref{lem:Holevo} since Alice obtains $Y_1$ conditioned on any realization of $({\sf B}, {\sf C})$ by measuring $\mathcal{Q}_2\mathcal{Q}_3$. Step \eqref{eq:ex_conv_2} is by Lemma \ref{lem:no_signal} conditioned on any realization of $({\sf B}, {\sf C})$, since ${\sf A}$ is not available to Server $3$. Step \eqref{eq:ex_conv_3} is because conditioning does not increase entropy.
By symmetry, one can similarly obtain,
\begin{align}
	L \leq H(\mathcal{Q}_2)_{\rho} - H(\mathcal{Q}_2 \mid \mathcal{Q}_1, {\sf A}, {\sf B}, {\sf C})_{\rho} \label{eq:ex_conv_4}
\end{align}
by considering the erasure of $\mathcal{Q}_3$ and conditioning on $({\sf A},{\sf B})$.
Adding \eqref{eq:ex_conv_3} and \eqref{eq:ex_conv_4}, we have
\begin{align}
	2L &\leq 2H(\mathcal{Q}_2)_{\rho} - \underbrace{\big(H(\mathcal{Q}_2 \mid \mathcal{Q}_3, {\sf A}, {\sf B}, {\sf C})_{\rho} + H(\mathcal{Q}_2 \mid \mathcal{Q}_1, {\sf A}, {\sf B}, {\sf C})_{\rho} \big)}_{\geq 0}\\
	&\leq 2H(\mathcal{Q}_2)_{\rho}
\end{align}
because of weak monotonicity,\footnote{For a tripartite quantum system $XYZ$ in the state $\rho$, $H(X|Y)_{\rho}+H(X|Z)_{\rho} \geq 0$.} conditioned on any realization of $({\sf A},{\sf B},{\sf C})$. This shows that $\Delta_2 = \log_d \delta_2/L$ $ \geq H(\mathcal{Q}_2)_{\rho}/L  \geq 1$. \hfill \qed

\subsection{$N$-sum box based  coding scheme for the example in Fig. \ref{fig:example} with $\Delta_0=0$}
Let us present the solution for the example in Fig. \ref{fig:example} for the case where Alice has no prior entanglement with the data-servers, i.e., $\Delta_0=0$ so Server $0$ can be ignored. Let  ${\sf A}_\ell, {\sf B}_\ell, \cdots, {\sf F}_\ell \in \mathbb{F}_5$ for all $\ell \in [L]$, and set $L=2$ as the batch size of the coding scheme. Thus, the coding scheme must allow Alice to compute ${\sf A}_{[2]}+{\sf B}_{[2]}+\cdots+{\sf F}_{[2]}$, with each server transmitting a  $\delta=5$ dimensional quantum system, while tolerating the loss of the quantum subsystem from any one of the servers. The rate to be achieved is $R = 2/(4\log_5 5) = 1/2$.
The scheme makes use of the $N$-sum box formulation in \cite{Allaix_N_sum_box}, which is summarized as a lemma below.
\begin{lemma}[$N$-sum box \cite{Allaix_N_sum_box}] \label{lem:box}
	Given a field $\mathbb{F}_q$, a positive integer $N$, and an $N\times 2N$ matrix ${\bf M} = [{\bf M}_l, {\bf M}_r]$ where ${\bf M}_l, {\bf M}_r \in \mathbb{F}_q^{N\times N}$ satisfy the strong self-orthogonality (SSO) property: $\rk({\bf M}) = N$, ${\bf M}_r{\bf M}_l^\top = {\bf M}_l{\bf M}_r^\top$, there exists a set of orthogonal quantum states, denoted as $\{\ket{{\bf a}}_{\bf M}\}_{{\bf a}\in \mathbb{F}_q^{N\times 1}}$ on $N$ $q$-dimensional quantum subsystems $Q_1,Q_2,\cdots, Q_N$, such that applying ${\sf X}(x_i){\sf Z}(z_i)$ to $Q_i$ for all $i \in [N]$, the state of the composite quantum system $Q$ changes from $\ket{\bf a}_{\bf M}$ to $\ket{{\bf a}+{\bf M}\bbsmatrix{{\bf x}\\{\bf z}}}_{\bf M}$ (with global phases omitted), i.e., $\bigotimes_{i\in [N]}{\sf X}(x_i){\sf Z}(z_i) \ket{{\bf a}}_{\bf M} \equiv \ket{{\bf a}+{\bf M}\bbsmatrix{{\bf x}\\{\bf z}}}_{\bf M}$, where ${\bf x} \triangleq [x_1,\cdots, x_N]^\top \in \mathbb{F}_q^{N\times 1}$ and ${\bf z} \triangleq [z_1,\cdots, z_N]^\top \in \mathbb{F}_q^{N\times 1}$. We say that $(\mathbb{F}_q, N, {\bf M}, \{\ket{{\bf a}}_{\bf M}\}_{{\bf a}\in \mathbb{F}_q^{N\times 1}})$ is an $N$-sum box defined in $\mathbb{F}_q$ with transfer matrix ${\bf M}$.
\end{lemma}
Lemma \ref{lem:box} follows essentially by a ``relabelling" of a set of orthogonal states that appear in the proof of \cite[Thm. 1]{Allaix_N_sum_box}.  For the sake of completeness, a proof is provided in Appendix \ref{proof:box}  along with some relevant background from \cite[Sec. IV-A]{song_colluding_PIR}, \cite[Sec. II]{Allaix_N_sum_box}.

Suppose the $5$-dimensional quantum subsystems available to the $4$ servers are $Q_1,Q_2,Q_3$ and $Q_4$, respectively. Let us construct an $N=4$-sum box defined in $\mathbb{F}_5$ with transfer matrix
\begin{align}
	{\bf M} &= \bbsmatrix{
		1 & 1 & 1 & 1 & 0 & 0 & 0 & 0 \\
		1 & 2 & 3 & 4 & 0 & 0 & 0 & 0 \\
		0 & 0 & 0 & 0 & 1 & 2 & 3 & 4 \\
		0 & 0 & 0 & 0 & 1 & 4 & 4 & 1
	} \\
	&= [{\bf m}_{1x},\cdots, {\bf m}_{4x}, {\bf m}_{1z}, \cdots, {\bf m}_{4z}]
\end{align}
where ${\bf m}_{ix}, {\bf m}_{iz}$ are the $i^{th}$, $(i+4)^{th}$ columns of ${\bf M}$ for $i\in [4]$. One can verify that this ${\bf M}$ satisfies the self-orthogonality constraint required by Lemma \ref{lem:box}. Moreover, this ${\bf M}$ has the following property.

\noindent {\bf Property P1:} Given any subset $\mathcal{I}\subseteq [4]$, the submatrix of ${\bf M}$ with columns ${\bf m}_{ix}, {\bf m}_{iz}$ for $i\in \mathcal{I}$ has full rank, equal to $\min\{4, 2|\mathcal{I}|\}$. 

For $s\in [4]$, Server $s$ applies ${\sf X}(x_s){\sf Z}(z_s)$ to its quantum subsystem, so that $[x_s, z_s]^\top$ is a linear function of the data that is available to Server $s$. Generally, we can write $[x_s,z_s]^\top$ as,
\begin{align}
	\bbsmatrix{ x_1\\z_1 } &=
	V_{1a}\bbsmatrix{{\sf A}_1\\{\sf A}_2} + V_{1b}\bbsmatrix{{\sf B}_1\\{\sf B}_2} + V_{1c}\bbsmatrix{{\sf C}_1\\{\sf C}_2} \\
	\bbsmatrix{ x_2\\z_2 } &=
	V_{2a}\bbsmatrix{{\sf A}_1\\{\sf A}_2} + V_{2d}\bbsmatrix{{\sf D}_1\\{\sf D}_2} + V_{2e}\bbsmatrix{{\sf E}_1\\{\sf E}_2} \\
	\bbsmatrix{ x_3\\z_3 } &=
	V_{3b}\bbsmatrix{{\sf B}_1\\{\sf B}_2} + V_{3d}\bbsmatrix{{\sf D}_1\\{\sf D}_2} + V_{3f}\bbsmatrix{{\sf F}_1\\{\sf F}_2} \\
	\bbsmatrix{ x_4\\z_4 } &=
	V_{4c}\bbsmatrix{{\sf C}_1\\{\sf C}_2} + V_{4e}\bbsmatrix{{\sf E}_1\\{\sf E}_2} + V_{4f}\bbsmatrix{{\sf F}_1\\{\sf F}_2} 
\end{align}
where $V_{s*}$ is a $2\times 2$ matrix with elements free to be chosen in $\mathbb{F}_5$, for all $s\in [4], *\in\{a,b,c,d,e,f\}$. According to Lemma \ref{lem:box}, if the initial state of the composite quantum system $Q$ is $\ket{{\bf 0}}_{\bf M}$, then after the operations, the state becomes $\ket{Y}_{\bf M}$ and
\begin{align}
	Y &= \bbsmatrix{{\bf m}_{1x} & {\bf m}_{1z} & {\bf m}_{2x} & {\bf m}_{2z}}\bbsmatrix{V_{1a}\\V_{2a}}\bbsmatrix{{\sf A}_1\\{\sf A}_2} \notag\\
	&+ \bbsmatrix{{\bf m}_{1x} & {\bf m}_{1z} & {\bf m}_{3x} & {\bf m}_{3z}}\bbsmatrix{V_{1b}\\V_{3b}}\bbsmatrix{{\sf B}_1\\{\sf B}_2} \notag\\
	&+ \bbsmatrix{{\bf m}_{1x} & {\bf m}_{1z} & {\bf m}_{4x} & {\bf m}_{4z}}\bbsmatrix{V_{1c}\\V_{4c}}\bbsmatrix{{\sf C}_1\\{\sf C}_2} \notag\\
	&+\bbsmatrix{{\bf m}_{2x} & {\bf m}_{2z} & {\bf m}_{3x} & {\bf m}_{3z}}\bbsmatrix{V_{2d}\\V_{3d}}\bbsmatrix{{\sf D}_1\\{\sf D}_2} \notag \\
	&+\bbsmatrix{{\bf m}_{2x} & {\bf m}_{2z} & {\bf m}_{4x} & {\bf m}_{4z}}\bbsmatrix{V_{2e}\\V_{4e}}\bbsmatrix{{\sf E}_1\\{\sf E}_2} \notag \\
	&+\bbsmatrix{{\bf m}_{3x} & {\bf m}_{3z} & {\bf m}_{4x} & {\bf m}_{4z}}\bbsmatrix{V_{3f}\\V_{4f}}\bbsmatrix{{\sf F}_1\\{\sf F}_2}\\
	& \triangleq {\bf M}_a {\bf V}_a \bbsmatrix{{\sf A}_1\\{\sf A}_2} + \cdots + {\bf M}_f {\bf V}_f \bbsmatrix{{\sf F}_1\\{\sf F}_2} \label{eq:simplify}
\end{align}
where in the last equation \eqref{eq:simplify} we define the compact notations ${\bf M}_*, {\bf V}_*$ for $*\in \{a,b,c,d,e,f\}$. The matrix ${\bf V}_*$ is $4\times 2$ with elements in $\mathbb{F}_5$ yet to be determined. It remains to specify ${\bf V}_*$. Let us first define a matrix,
\begin{align}
	{\bf U} = \bbsmatrix{4 & 3\\1 & 1 \\ 2& 2 \\1 & 3}.
\end{align}
There is nothing too special about this choice of ${\bf U}$, except that we need it to satisfy the following property.

\noindent {\bf Property P2:} $\rk([{\bf U}, {\bf m}_{tx}, {\bf m}_{tz}])=4$, i.e., $[{\bf U}, {\bf m}_{tx}, {\bf m}_{tz}]$ is invertible for $t\in [4]$.

By Property {\bf P1}, for $*\in \{a,b,c,d,e,f\}$, the matrix ${\bf M}_*$ has full rank $4$ and is invertible. Now we can specify ${\bf V}_*$ such that
\begin{align}
	Y &= {\bf U} \bbsmatrix{{\sf A}_1\\{\sf A}_2} + {\bf U} \bbsmatrix{{\sf B}_1\\{\sf B}_2} +\cdots + {\bf U} \bbsmatrix{{\sf F}_1\\{\sf F}_2} \\
	& = {\bf U}\Big( \bbsmatrix{{\sf A}_1\\{\sf A}_2} + \bbsmatrix{{\sf B}_1\\{\sf B}_2} + \cdots + \bbsmatrix{{\sf F}_1\\{\sf F}_2} \Big)
\end{align}
by letting ${\bf V}_* = {\bf M}_*^{-1}{\bf U}$ for $* \in \{a,b,c,d,e,f\}$.

The next step is to consider the transmission of the quantum subsystems. Suppose during the transmission, there exists one $t\in [4]$ such that $Q_t$ is subjected to another operation ${\sf X}(\tilde{x}){\sf Z}(\tilde{z})$ where $\tilde{x}, \tilde{z}$ are unknown. The state of the composite system now becomes $\ket{Y'}_{\bf M}$ and 
\begin{align}
	Y' = {\bf U}\Big( \bbsmatrix{{\sf A}_1\\{\sf A}_2} + \bbsmatrix{{\sf B}_1\\{\sf B}_2} + \cdots + \bbsmatrix{{\sf F}_1\\{\sf F}_2} \Big)  +  \bbsmatrix{{\bf m}_{tx}, {\bf m}_{tz}} \bbsmatrix{\tilde{x}\\\tilde{z}}.
\end{align}
After receiving the composite system, Alice obtains $Y'$ as the result of her measurement with certainty. By Property {\bf P2}, if Alice knows $t$, then she can retrieve
\begin{align}
	\bbsmatrix{{\sf A}_1\\{\sf A}_2} + \bbsmatrix{{\sf B}_1\\{\sf B}_2} + \cdots + \bbsmatrix{{\sf F}_1\\{\sf F}_2}
\end{align}
by multiplying $[{\bf U}, {\bf m}_{tx}, {\bf m}_{tz}]^{-1}$ to $Y'$ and taking the top $2$ elements.  Note that the scheme works no matter which quantum subsystem $Q_t$ is subjected to the unknown operations ${\sf X}(\tilde{x}){\sf Z}(\tilde{z})$.  In addition, $\tilde{x}, \tilde{z}$ can be any value in $\mathbb{F}_5$. Therefore, the scheme also works if $\tilde{x}, \tilde{z}$ are independent and uniformly drawn in $\mathbb{F}_5$. But if $Q_t$ is subjected to the operation ${\sf X}(\tilde{x}){\sf Z}(\tilde{z})$ with independent $\tilde{x}, \tilde{z}$ chosen uniformly in $\mathbb{F}_5$, then it  puts $Q_t$ in the maximally mixed state and makes it \emph{independent} of (and unentangled with) the rest of the quantum subsystems (Lemma \ref{lem:independent_mixed}). In other words, Alice is able to recover the desired sum even if any one of the $4$ quantum systems (Alice knows which one) is subjected to random ${\sf X}(\tilde{x}){\sf Z}(\tilde{z})$ operations that make it independent of the rest of the quantum systems. Recall that positions of erasures are by definition known to the receiver.

What this implies is that even if $Q_t$ is lost, i.e., not received by Alice, the scheme still works \cite{QECtutorial} if Alice simply replaces the missing $\mathcal{Q}_t$ by an ancillary $5$-dimensional quantum subsystem $\widehat{Q}$ that is in the maximally mixed state described by the density operator $\rho_{\widehat{Q}} = {\bf I}_5/5$, where ${\bf I}_5$ denotes the $5\times 5$ identity matrix, and $\widehat{Q}$ is independent of the quantum system $Q$ (with density operator $\rho_Q$) composed of $Q_1,Q_2,Q_3,Q_4$, i.e., the joint state of $Q$ and $\widehat{Q}$ is $\rho_{Q \widehat{Q}} = \rho_Q \otimes \rho_{\widehat{Q}}$.
\begin{lemma} \label{lem:independent_mixed}
	Let $A,B$ denote two quantum subsystems in the joint state $\rho_{AB}$ where $B$ has dimension $d=p^r$ with $d$ being a power of a prime. Suppose $B$ is subjected to a random operation ${\sf X}(\tilde{x}){\sf Z}(\tilde{z})$ with independent $\tilde{x}, \tilde{z}$ drawn uniformly in $\mathbb{F}_d$. Denote the joint state, the partial state of $A$, and the partial state of $B$ after the random operation as $\rho'_{AB}, \rho_A$ and $\rho_B'$, respectively. Then $\rho'_{AB} = \rho_A \otimes \rho_{B}'$ and $\rho_B' = {\bf I}_d/d$, where ${\bf I}_d$ denotes the $d\times d$ identity matrix.
\end{lemma}
The argument in the lemma is  standard in quantum literature  \cite{QECtutorial}, \cite[Exercise 4.7.6 (Qudit Twirl)]{Wilde_2017}. A proof is provided in Appendix \ref{proof:independent_mixed}.

\section{Proof of Theorem \ref{thm:achievability}} \label{proof:achievability}
\subsection{Scheme with download cost tuple in $\mathfrak{D}_{\normalfont \footnotesize \mbox{AME}}$} \label{sec:scheme_construction}
Let us design a scheme which allows Alice to compute $L=\lambda l$ instances of the sum $\sum_{k\in[K]}{\sf W}_k^{[L]}$ $\in \mathbb{F}_d^{L\times 1}$. To do so, let the  $S+1$ servers prepare $N=N_0+N_1+N_2+\cdots+N_S$ $q$-dimensional quantum subsystems, $Q_0,Q_1,Q_2,\cdots Q_N$, with $q = d^{\lambda}$. For $s\in \mathcal{S} = \{0,1,\cdots, S\}$, Server $s$ possesses $N_s$ of these quantum subsystems indexed by $\mathcal{I}_s \subseteq [N]$. Let ${\bf M}\in \mathbb{F}_q^{N\times 2N}$ be the transfer matrix of an $N$-sum box where $N = \sum_{s\in\mathcal{S}}N_s$. By definition, an element of $\mathbb{F}_{q}$ is equivalent to a $\lambda$-length vector in $\mathbb{F}_d$. Thus, let ${\bf W}_k \in \mathbb{F}_q^{l\times 1}$ denote ${\sf W}_k^{[L]}$ written in $\mathbb{F}_q$ for $k\in [K]$. Recall that Server $s$ possesses the $N_s$ $q$-dimensional quantum subsystems $Q_i$ for $i\in \mathcal{I}_s$. 

Define $x_i=0, z_i=0,$ for all $i\in\mathcal{I}_0$, since the auxiliary server has no data inputs. Let Server $s\in\mathcal{S}$  apply the operation ${\sf X}(x_{i}){\sf Z}(z_{i})$ to $Q_i$ for $i \in \mathcal{I}_s$. Suppose that each $x_i, z_i$ is equal to the output of a linear function of the data available to the server. If the initial state of the composite system is $\ket{\bf 0}_{\bf M}$, then after each server applies the operations, the output state  can be written as
\begin{align} \label{eq:output_state_1}
	\ket{{\bf M}\bbsmatrix{{\bf x} \\ {\bf z}}}_{\bf M} = \Ket{\sum_{k\in [K]}{\bf M}_k{\bf V}_k {\bf W}_k}_{\bf M},
\end{align}
where for $k\in [K]$, ${\bf V}_k$ is a $\sum_{s\in \mathcal{W}(k)} 2N_s \times l$ matrix with elements freely chosen in $\mathbb{F}_q$, and ${\bf M}_k$ is the submatrix composed of the columns of ${\bf M}$ `controlled' by Servers $s\in \mathcal{W}(k)$. Specifically, for $k\in [K]$, ${\bf M}_k$ contains the $i^{th}$ and $(i+N)^{th}$ columns of ${\bf M}$ for $i\in \cup_{s\in \mathcal{W}(k)}\mathcal{I}_s$ and thus ${\bf M}_k$ has size $N\times \sum_{s\in \mathcal{W}(k)}2N_s$. It is proved in \cite{Yao_Jafar_Sum_MAC} that if $q \geq N$, then there exists an ${\bf M}$ (as the transfer matrix of an $N$-sum box) such that $\rk({\bf M}_k) = \min\{N, \sum_{s\in \mathcal{W}(k)}2N_s\}$ for all $k\in [K]$. 

Recall that during the transmission, the quantum subsystems possessed by Servers $s\in \mathcal{E}(t)$ may be erased for any  $t\in [T]$. To tolerate this erasure, let us find a matrix ${\bf U}\in \mathbb{F}_q^{N\times u}$ with $u =  N-\max_{t\in [T]} \sum_{s\in \mathcal{E}(t)}2N_s$ that has full column rank $u$. Here, we require that $N \geq  \max_{t\in [T]}\sum_{s\in \mathcal{E}(t)}2N_s$ for the scheme to work. Further, for $t\in [T]$, define ${\bf E}_t$ as the submatrix of ${\bf M}$ composed by the columns indexed by $i, i+N$ for $i\in \cup_{s\in \mathcal{E}(t)}\mathcal{I}_s$. We require that $[{\bf U}, {\bf E}_t]$ must have full rank $u+\sum_{s\in \mathcal{E}(t)}2N_s$ for all $t\in [T]$. The existence of such a ${\bf U}$ is guaranteed if $q>TN$ (proof provided in Appendix \ref{proof:schwartz_zippel}).
In the following, let $\langle {\bf A} \rangle$ denote the linear subspace spanned by the columns of ${\bf A}$ with coefficients chosen in $\mathbb{F}_q$. For $k\in [K]$, denote ${\bf U}_k$ as a basis of the linear subspace spanned by the intersection of $\langle{\bf U}\rangle$ and $\langle{\bf M}_k\rangle$. 
For $k\in [K]$, since $\langle {\bf U}_k \rangle \subseteq \langle {\bf M}_k \rangle$, it follows from basic linear algebra that for any matrix ${\bf V}_k' \in \mathbb{F}_q^{\scalebox{0.8}{\rk}({\bf U}_k) \times l}$, there exists ${\bf V}_k\in \mathbb{F}_q^{\sum_{s\in \mathcal{W}(k)} 2N_s \times l}$ such that ${\bf M}_k {\bf V}_k = {\bf U}_k {\bf V}_k'$.
Then the output state \eqref{eq:output_state_1} of the composite quantum system can be written as
\begin{align} \label{eq:output_state_2}
	\Ket{\sum_{k\in [K]}{\bf U}_k {\bf V}_k' {\bf W}_k}_{\bf M},
\end{align}
where ${\bf V}_k' \in \mathbb{F}_q^{\scalebox{0.8}{\rk}({\bf U}_k) \times l}$ has full column rank, equal to $l$ if $l\leq \rk({\bf U}_k), \forall k\in [K]$.
By the dimension law for linear subspaces, 
\begin{align}
	\rk({\bf U}_k) &= \rk({\bf U})+\rk({\bf M}_k) -\rk([{\bf U}, {\bf M}_k]) \\
	&\geq u+\min\{N, \sum_{s\in \mathcal{W}(k)}2N_s\}-N \\
	&= \min\{N, \sum_{s\in \mathcal{W}(k)}2N_s\} - \max_{t\in [T]}\sum_{s\in \mathcal{E}(t)}2N_s
\end{align}
for all $k\in [K]$. 
After the servers apply the operations to their quantum subsystems, suppose in the transmission,  the operations ${\sf X}(\tilde{x}_i){\sf Z}(\tilde{z}_i)$ are applied to the $i^{th}$ quantum subsystem, with unknown $\tilde{x}_i, \tilde{z}_i$ for $i\in \cup_{s\in \mathcal{E}(t)}$. The state received by Alice becomes
\begin{align}\label{eq:output_state_3}
	\Ket{\sum_{k\in [K]} {\bf U}_k {\bf V}_k' {\bf W}_k + {\bf E}_t \bbsmatrix{\tilde{x}_{i_1}\\ \vdots \\ \tilde{x}_{i_n} \\\tilde{z}_{i_1} \\\vdots \\\tilde{z}_{i_n}}}_{\bf M}
\end{align}
where $\{i_1,\cdots, i_n\} = \cup_{s\in \mathcal{E}(t)} \mathcal{I}_s$. Alice can now measure the result
\begin{align}
	Y = \underbrace{\sum_{k\in [K]}{\bf U}_k {\bf V}_k' {\bf W}_k}_{{\bf v}_u} + \underbrace{{\bf E}_t \bbsmatrix{\tilde{x}_{i_1}\\ \vdots \\ \tilde{x}_{i_n} \\\tilde{z}_{i_1} \\\vdots \\\tilde{z}_{i_n}}}_{{\bf v}_e}
\end{align}
with certainty given any $t\in [T]$. Since for $k\in [K]$, $\langle {\bf U}_k \rangle \subseteq \langle {\bf U} \rangle$, we have ${\bf v}_{u} \in \langle {\bf U} \rangle$. Similarly, ${\bf v}_e \in \langle {\bf E}_t\rangle$. Meanwhile, since for any $t\in [T]$, $[{\bf U}, {\bf E}_t]$ has full column rank, Alice (who knows $t$ and $Y$) is then able to compute ${\bf v}_u$ and ${\bf v}_e$. One way to do this is to first solve for ${\bf c}_u, {\bf c}_e$ such that ${\bf v}_u = {\bf U} {\bf c}_u, {\bf v}_e = {\bf E}_t {\bf c}_e$, by multiplying the left-inverse of $[{\bf U}, {\bf E}_t]$ to $Y$, and then compute ${\bf v}_u$ and ${\bf v}_e$ accordingly. Omitting the noise term ${\bf v}_e$, Alice now has ${\bf v}_u$, which is
\begin{align} \label{eq:general_achi_decoding}
	\sum_{k\in [K]} {\bf U}_k {\bf V}_k' {\bf W}_k \in \mathbb{F}_q^{N\times 1}
\end{align}
and further compute
\begin{align}
	{\bf V}_{\dec} \sum_{k\in [K]} {\bf U}_k {\bf V}_k' {\bf W}_k = \sum_{k\in [K]} {\bf R}_k {\bf W}_k \in \mathbb{F}_q^{l\times 1}
\end{align}
where ${\bf R}_k$ is an $l\times l$ invertible square matrix for $k\in [K]$. The existence of the ${\bf V}_{\dec}$ is guaranteed if $q> Kl$ (proof provided in Appendix \ref{proof:schwartz_zippel}). Since ${\bf R}_k$ is invertible for all $k\in [K]$, the servers can simply treat ${\bf W}_k$ as ${\bf R}_k^{-1}{\bf W}_k$ when coding so that Alice is able to compute the desired sum $\sum_{k\in [K]}{\bf W}_k$ in the end.

\subsection{Extend the scheme to tolerate erasures}
In the analysis of the scheme, we assumed that the quantum subsystems for Servers $s\in \mathcal{E}(t)$ are subjected to unknown ${\sf X}$ and ${\sf Z}$ operations during the transmission. If  these quantum subsystems are lost during the transmission, a simple solution is to replace the lost quantum systems with ancillary quantum subsystems generated locally by Alice, which are independent of the quantum systems sent from the servers. The reasoning follows from Lemma \ref{lem:independent_mixed}. Thus, the scheme in Section \ref{sec:scheme_construction} now tolerates erasures. Note that such a scheme exists if
\begin{align}
	\left\{
	\begin{array}{l}
		l \in \mathbb{N}, \\
		\Delta_s = N_s/l, \forall s\in \mathcal{S},\\
		\min\{N, \sum_{s\in \mathcal{W}(k)}2N_s\} - \max_{t\in \mathcal{E}(t)}\sum_{s\in \mathcal{E}(t)}2N_s \geq l,\\ \forall k, t\in [T].
	\end{array}
	\right.
\end{align}
Then by definition,
\begin{align} \label{eq:closure_AME}
	& \mbox{closure} \left\{ 
	{\bf \Delta} \in \mathbb{R}_+^{S}
    \left|
    \begin{array}{l}
      l \in \mathbb{N}, \\
		\Delta_s = N_s/l, \forall s\in \mathcal{S},\\
		\min\{N, \sum_{s\in \mathcal{W}(k)}2N_s\} \\- \max_{t\in \mathcal{E}(t)}\sum_{s\in \mathcal{E}(t)}2N_s \geq l, \forall k, t\in [T]
    \end{array}
    \right. \hspace{-0.2cm}
    \right\}   \subseteq \mathfrak{D}^*.
\end{align}
Note that the LHS of \eqref{eq:closure_AME} is equal to $\mathfrak{D}_{\AME}$. Therefore, $\mathfrak{D}_{\AME} \subseteq \mathfrak{D}^*$. It then follows that $\mathfrak{D}_{\AME}^+ \subseteq \mathfrak{D}^*$ as one can always send extra unentangled quantum resource.

\subsection{Treating Qudits as Classical-dits (TQC)}
Let us show that $\mathfrak{D}_{\TQC} \subseteq \mathfrak{D}^*$. Suppose for $s\in \mathcal{S}$, Server $s$ sends $N_s$ $q$-dimensional quantum subsystems by treating them as classical $q$-ary symbols, where $q=d^\lambda$. Denote $N \triangleq N_0+N_1+\cdots+N_S$. During the transmission, the quantum subsystems (now regarded as classcial $q$-ary symbols) from servers $s\in \mathcal{E}(t)$ may be lost. Alice should be able to compute $l$ instances of the sum, written in $\mathbb{F}_q$ as $\sum_{k\in [K]}{\bf W}_k$, where ${\bf W}_k \in \mathbb{F}_q^{l\times 1}$. Let us restate the result of \cite{wei2023robust} in the following lemma.
\begin{lemma}[Restatement of \cite{wei2023robust}]
	The classical scheme exists if $q>N$, and 
	\begin{align}
		\max_{t\in [T]}\sum_{s \in \mathcal{E}(t)}N_s \leq \min_{k\in [K]} \sum_{s\in \mathcal{W}(k)}N_s -l.
	\end{align}
\end{lemma}
Since $q>N$ can always be satisfied by choosing large enough $\lambda$,  we have by definition,
\begin{align} \label{eq:closure_TQC}
	\mbox{closure} \left\{ 
	{\bf \Delta} \in \mathbb{R}_+^{S}
    \left|
    \begin{array}{l}
        l \in \mathbb{N}, \\
        \Delta_s = N_s/l, \forall s\in [S],\\
		\max_{t\in [T]}\sum_{s \in \mathcal{E}(t)}N_s \\ \leq \min_{k\in [K]} \sum_{s\in \mathcal{W}(k)}N_s -l
    \end{array}
    \right.
    \right\} \subseteq \mathfrak{D}^*.
\end{align}
Note that the LHS of \eqref{eq:closure_TQC} is equal to $\mathfrak{D}_{\TQC}$. Therefore, $\mathfrak{D}_{\TQC} \subseteq \mathfrak{D}^*$. 
It then follows that $\mathfrak{D}_{\TQC}^+ \subseteq \mathfrak{D}^*$ as one can always send extra unentangled quantum resource. Finally, by a time-sharing argument, any convex combination of the tuples in $\mathfrak{D}_{\AME}^+ \cup \mathfrak{D}_{\TQC}^+ \subseteq \mathfrak{D}^*$. This concludes the proof of Theorem \ref{thm:achievability}. \hfil \qed

\section{Proof of Theorem \ref{thm:sym}} \label{proof:sym}
\subsection{Achievability}
Apply Theorem \ref{thm:achievability} to the $\Sigma^s$-QEMAC setting, with the cost tuple in the form ${\bf \Delta} = (\Delta_0, \Delta, \cdots, \Delta)$. From \eqref{eq:region_AME} we have  ${\bf \Delta} \in \mathfrak{D}_{\AME}$ if 
\begin{align} \label{eq:sym_AME}
	\min\{\Delta_0+S\Delta, 2\alpha \Delta \} - 2\beta\Delta \geq 1.
\end{align}
Similarly \eqref{eq:region_TQC} implies that  ${\bf \Delta} \in \mathfrak{D}_{\TQC}$ if 
\begin{align} \label{eq:sym_TQC}
	(\alpha-\beta)\Delta \geq 1.
\end{align}
Consider two cases.
\begin{enumerate}[align=left, leftmargin = 1cm]
	\item[{\bf Case 1:}]  $S \geq \alpha+\beta$. It follows from $\alpha > \beta$ that $S-2\beta>0$ and from \eqref{eq:sym_AME} that for any $\Delta_0\in \mathbb{R}_+$ and $\Delta = \max\{ \frac{1}{2(\alpha-\beta)}, \frac{1-\Delta_0}{S-2\beta} \}$, we have $(\Delta_0, \Delta,\cdots, \Delta) \in \mathfrak{D}_{\AME} \subseteq \mathfrak{D}^*$. 
	\item[{\bf Case 2:}] $S < \alpha+\beta$. It follows from \eqref{eq:sym_AME} that for $\big(\Delta_0 = \frac{2\alpha-S}{2(\alpha-\beta)}, \Delta = \frac{1}{2(\alpha-\beta)}\big)$ the cost tuple $(\Delta_0, \Delta,\cdots, \Delta)$ $\in \mathfrak{D}_{\AME}$. Also \eqref{eq:sym_TQC} implies that for $\big(\Delta_0 = 0, \Delta = \frac{1}{\alpha-\beta}\big)$ the cost tuple $(\Delta_0, \Delta,\cdots, \Delta) \in \mathfrak{D}_{\TQC}$. It can be verified that the equation for the straight line connecting these two points is $\Delta = \frac{1}{\alpha-\beta}-\frac{\Delta_0}{2\alpha-S}$. It thus follows that for any $\Delta_0 \in \mathbb{R}_+$ and $\Delta = \max\{ \frac{1}{2(\alpha-\beta)}, \frac{1}{\alpha-\beta}-\frac{\Delta_0}{2\alpha-S}\}$, we have $(\Delta_0, \Delta,\cdots, \Delta)$ $ \in \conv (\mathfrak{D}_{\AME}^+ \cup \mathfrak{D}_{\TQC}^+) \subseteq \mathfrak{D}^*$.
\end{enumerate}
Given $\Delta_0 \in \mathbb{R}_+$, by definition,  $C(\Delta_0) \geq (S\Delta)^{-1}$ if $(\Delta_0, \Delta, \cdots, \Delta) \in \mathfrak{D}^*$. Thus, $C(\Delta_0) \geq (S\Delta^*)^{-1}$ with the $\Delta^*$ defined in \eqref{eq:sym_optimal_cost}. \hfil \qed

\subsection{Converse} \label{proof:thm_conv}
Recall that the  $\Sigma^s$-QEMAC contains $K = \binom{S}{\alpha}$ data streams and $S+1$ servers. There are $T = \binom{S}{\beta}$ possible cases of erasure, each corresponding to a unique $\beta$ subset of the answers from the $S$ data-servers being erased. For any scheme with batch size $L$, since the decoding must be correct for all data realization $w=(w_1,w_2,\cdots,w_K) \in \mathbb{F}_d^{KL}$, it must be correct if the $K$ data streams are uniformly distributed over $\mathbb{F}_d^{KL}$. Therefore, for the purpose of proving a converse, it will not hurt to assume that the data is uniformly distributed over $\mathbb{F}_d^{KL}$, and let $({\sf W}_1, {\sf W}_2,\cdots, {\sf W}_K)$ denote these $K$ random variables, one for each data stream. This facilitates entropic proofs.
After the coding operations applied by the servers, the classical-quantum system of interest 
\begin{align*}
	{\sf W}_1\cdots {\sf W}_K \mathcal{Q}_0\mathcal{Q}_1 \cdots\mathcal{Q}_S 
\end{align*}
is in the state
\begin{align}
	\rho_{{\sf W}_1\cdots {\sf W}_K \mathcal{Q}_0\mathcal{Q}_1\cdots \mathcal{Q}_S} = \sum_{w\in  \mathbb{F}_d^{KL}}\frac{1}{d^{KL}}\rho_{\mathcal{Q}_0\mathcal{Q}_1\cdots \mathcal{Q}_S}^{(w)}.
\end{align}
Also recall that for the case of erasure indexed by $t\in [T]$, the measurement result is a random variable, denoted as $Y_{t}$ such that
\begin{align}
	{\sf Pr}(Y_t = y \mid {\sf W}_{[K]} = w) = \Tr(\Lambda_t^{(y)}\rho_{\mathcal{Q}_0 \mathcal{Q}_{[S]\setminus \mathcal{E}(t)}}),
\end{align}
where $\{\Lambda_t^{(y)}\}_y$ is the set of POVM matrices associated with the measurement $\Psi_{t}$. The joint distribution of ${\sf W}_1\cdots {\sf W}_K Y_t$ is thus determined.
According to the definition \eqref{eq:def_correctness}, the scheme must have ${\sf Pr}\left(Y_{t} = \sum_{k\in [K]}{\sf W}_k \right) = 1$ for $t\in [T]$. For $s\in \mathcal{S}=\{0,1,\cdots, S\}$, the dimension of $\mathcal{Q}_s$ is $\delta_s$ and due to symmetry there is no loss of generality in considering schemes with $\delta_1=\delta_2=\cdots=\delta_S = \delta$.

The proof proceeds as follows. We shall show that for any $\Delta_0 \in \mathbb{R}_+$, if $(\Delta_0, \Delta, \cdots, \Delta) \in \mathfrak{D}^*$, then
\begin{align}
	\Delta \geq \frac{1}{2(\alpha-\beta)} 
\end{align}
and in addition,
\begin{align}
	& \Delta \geq \frac{1-\Delta_0}{S-2\beta}, ~~\mbox{if}~ S\geq \alpha +\beta, \\
	& \Delta \geq \frac{1}{\alpha-\beta} - \frac{\Delta_0}{2\alpha-S}, ~~ \mbox{if}~ S < \alpha +\beta.
\end{align} 
First, consider any data stream (say ${\sf W}_1$), and condition on any realization of ${\sf W}_{[K] \setminus \{1\}}$, the problem reduces to a communication problem with one data stream ${\sf W}_1$, such that $\alpha$ servers (Servers $s\in \mathcal{W}(1)$) know the data stream, while the rest of the servers merely provide entanglement to Alice. Theorem \ref{thm:eacq_bound} (Eq. \eqref{eq:eacq_1}) then implies that $(\alpha-\beta)\Delta \geq 1/2 \implies \Delta \geq \frac{1}{2(\alpha-\beta)}$ (by substituting $S$ by $\alpha$).

Now, consider the following cases.
\begin{enumerate}[align = left]
	\item[{\bf Case 1:}] $S\geq \alpha+\beta$. Partition $[S] = \{1,2,\cdots, S\}$ such that
		\begin{align}
			[S] = &\underbrace{\{1,2,\cdots, \beta\}}_{\mathcal{I}_1~ (\beta)} \cup \underbrace{\{\beta+1,\cdots,2\beta\}}_{\mathcal{I}_2 ~ (\beta)} \cup \underbrace{\{2\beta+1,\cdots, \alpha+\beta\}}_{\mathcal{I}_3 ~ (\alpha-\beta)} \cup \underbrace{\{\alpha+\beta+1,\cdots, S\}}_{\mathcal{I}_4~ (S-\alpha-\beta)}.
		\end{align}
		Identify indices $k_1, k_2, t_1, t_2$ such that
		\begin{align}
			&\mathcal{W}(k_1) = \underbrace{\mathcal{I}_1\cup \mathcal{I}_3}_{(\alpha)}, &&\mathcal{W}(k_2) = \underbrace{\mathcal{I}_2 \cup \mathcal{I}_3}_{(\alpha)},\notag \\
			&\mathcal{E}(t_1) = \underbrace{\mathcal{I}_1}_{(\beta)}, && \mathcal{E}(t_2) = \underbrace{\mathcal{I}_2}_{(\beta)}.
		\end{align}
		For the case of erasure indexed by $t_1$, i.e., when $\mathcal{Q}_{\mathcal{I}_1}$ is erased, we have
		\begin{align}
			L &= I\Big({\sf W}_{[K]} ; \sum_{k\in [K]}{\sf W}_k \mid {\sf W}_{[K]\setminus \{k_1\}} \Big) \\
			& \leq I\Big({\sf W}_{[K]} ; Y_{t_1} \mid {\sf W}_{[K]\setminus \{k_1\}} \Big)\\
			& \leq I\Big( {\sf W}_{k_1}; \mathcal{Q}_{0} \mathcal{Q}_{\mathcal{I}_2} \mathcal{Q}_{\mathcal{I}_3} \mathcal{Q}_{\mathcal{I}_4} \mid {\sf W}_{[K]\setminus \{k_1\}} \Big)  \label{eq:conv_1_1} \\
			& = \underbrace{I\Big( {\sf W}_{k_1};   \mathcal{Q}_{\mathcal{I}_2}   \mid {\sf W}_{[K]\setminus \{k_1\}} \Big)}_{=0}  +  I\Big( {\sf W}_{k_1};  \mathcal{Q}_{0} \mathcal{Q}_{\mathcal{I}_3} \mathcal{Q}_{\mathcal{I}_4} \mid  \mathcal{Q}_{\mathcal{I}_2},  {\sf W}_{[K]\setminus \{k_1\}} \Big)  \\
			& = I\Big( {\sf W}_{k_1}; \mathcal{Q}_{0} \mathcal{Q}_{\mathcal{I}_3} \mathcal{Q}_{\mathcal{I}_4} \mid \mathcal{Q}_{\mathcal{I}_2},  {\sf W}_{[K]\setminus \{k_1\}} \Big) \label{eq:conv_1_2} \\
			& = H\Big(\mathcal{Q}_{0} \mathcal{Q}_{\mathcal{I}_3} \mathcal{Q}_{\mathcal{I}_4} \mid \mathcal{Q}_{\mathcal{I}_2},  {\sf W}_{[K]\setminus \{k_1\}} \Big) - H\Big(\mathcal{Q}_{0} \mathcal{Q}_{\mathcal{I}_3} \mathcal{Q}_{\mathcal{I}_4} \mid \mathcal{Q}_{\mathcal{I}_2},  {\sf W}_{[K]} \Big)  \\
			& \leq H(\mathcal{Q}_{0} \mathcal{Q}_{\mathcal{I}_3} \mathcal{Q}_{\mathcal{I}_4}) - H\Big(\mathcal{Q}_{0} \mathcal{Q}_{\mathcal{I}_3} \mathcal{Q}_{\mathcal{I}_4} \mid \mathcal{Q}_{\mathcal{I}_2},  {\sf W}_{[K]} \Big) \\
			& \leq \log_d\delta_0+ (S-2\beta) \log_d \delta  - H\Big(\mathcal{Q}_{0} \mathcal{Q}_{\mathcal{I}_3} \mathcal{Q}_{\mathcal{I}_4} \mid \mathcal{Q}_{\mathcal{I}_2},  {\sf W}_{[K]} \Big) \label{eq:conv_1_3}
		\end{align}
		Information measures on and after Step \eqref{eq:conv_1_1} are with respect to the state $\rho_{{\sf W}_1\cdots {\sf W}_K \mathcal{Q}_0 \mathcal{Q}_1\cdots \mathcal{Q}_S}$. Step \eqref{eq:conv_1_1} is by Lemma \ref{lem:Holevo} since Alice obtains $Y_{t_1}$ conditioned on any realization of ${\sf W}_{[K]\setminus \{k_1\}}$ by measuring $\mathcal{Q}_0 \mathcal{Q}_{\mathcal{I}_2}\mathcal{Q}_{\mathcal{I}_3}\mathcal{Q}_{\mathcal{I}_4}$. Step \eqref{eq:conv_1_2} is by Lemma \ref{lem:no_signal} conditioned on any realization of ${\sf W}_{[K]\setminus \{k_1\}}$, since ${\sf W}_{k_1}$ is not available to any server in $\mathcal{I}_2$. Step \eqref{eq:conv_1_3} is because conditioning does not increase entropy \cite[Thm. 11.4.1]{Wilde_2017}.
		
		For the case of erasure indexed by $t_2$, by the same reasoning, we similarly have ($\mathcal{I}_2$ replaced by $\mathcal{I}_1$),
		\begin{align}
			L &\leq \log_d\delta_0+ (S-2\beta) \log_d \delta  - H\Big(\mathcal{Q}_{0} \mathcal{Q}_{\mathcal{I}_3} \mathcal{Q}_{\mathcal{I}_4} \mid \mathcal{Q}_{\mathcal{I}_1},  {\sf W}_{[K]} \Big). \label{eq:conv_1_4}
		\end{align}
		Adding \eqref{eq:conv_1_3} and \eqref{eq:conv_1_4}, we have
		\begin{align}
			2L &\leq 2\big(\log_d\delta_0+ (S-2\beta) \log_d \delta\big) \notag\\
			&~~~~- \Bigg( H\Big(\mathcal{Q}_{0} \mathcal{Q}_{\mathcal{I}_3} \mathcal{Q}_{\mathcal{I}_4} \mid \mathcal{Q}_{\mathcal{I}_2},  {\sf W}_{[K]} \Big)  + H\Big(\mathcal{Q}_{0} \mathcal{Q}_{\mathcal{I}_3} \mathcal{Q}_{\mathcal{I}_4} \mid \mathcal{Q}_{\mathcal{I}_1},  {\sf W}_{[K]} \Big) \Bigg)  \\
			& \leq 2\big(\log_d\delta_0+ (S-2\beta) \log_d \delta\big) \label{eq:conv_1_5} \\
			& \hspace{-1cm} \implies \Delta_0 + (S-2\beta) \Delta \geq 1 \label{eq:conv_1_6} \\
			& \hspace{-1cm} \implies  \Delta \geq  \frac{1 - \Delta_0}{S-2\beta}. 
		\end{align}
		Step \eqref{eq:conv_1_5} follows from  weak monotonicity\footnote{For a tripartite quantum system $ABC$ in the state $\rho$, $H(A|B)_{\rho}+H(A|C)_{\rho} \geq 0$.}\cite{pippenger2003inequalities, linden2005new} conditioned on ${\sf W}_{[K]} = w$ for any $w\in \mathbb{F}_d^{KL}$. Step \eqref{eq:conv_1_6} is by dividing by $2L$ on both sides and using the definition $\Delta_0 = \log_d \delta_0/L$, $\Delta =  \log_d \delta/L$.
		\item[{\bf Case 2:}] $S < \alpha+\beta$. The idea is similar but the proof requires a few more steps. Partition $[S]$ such that,
		
		\begin{align}
			[S] =& \underbrace{\{1,2,\cdots, 2\alpha -S\}}_{\mathcal{I}_1~ (2\alpha-S)} \cup \underbrace{\{2\alpha-S+1, \cdots, \alpha\}}_{\mathcal{I}_2~(S-\alpha)} \cup \underbrace{\{\alpha +1, \cdots, S\}}_{\mathcal{I}_3~ (S-\alpha)}.
		\end{align}
		Note that since $S<\alpha+\beta$, we have $2\alpha -S > \alpha -\beta$. Let $\mathcal{J} \in \binom{[2\alpha-S]}{\alpha-\beta}$ be any subset of $\mathcal{I}_1$ with cardinality equal to $\alpha-\beta$. Identify indices $k, t$ such that 
			\begin{align}
				&\mathcal{W}(k) = \underbrace{\mathcal{I}_1 \cup \mathcal{I}_3}_{(\alpha)}, \\
				&\mathcal{E}(t) = \underbrace{(\mathcal{I}_1 \setminus \mathcal{J}) \cup \mathcal{I}_3}_{(\beta)}.
			\end{align}
		We have
		\begin{align}
			L &= I\big({\sf W}_{[K]}; \sum_{k\in [K]} {\sf W}_k \mid {\sf W}_{[K]\setminus \{k\}} \big) \\
			&\leq I\big({\sf W}_{[K]}; Y_{t} \mid {\sf W}_{[K]\setminus \{k\}} \big) \\
			&\leq I\big({\sf W}_{k}; \mathcal{Q}_0 \mathcal{Q}_{\mathcal{J}} \mathcal{Q}_{\mathcal{I}_2} \mid {\sf W}_{[K]\setminus \{k\}} \big)  \label{eq:conv_2_1} \\
			&= \underbrace{I\big({\sf W}_{k};  \mathcal{Q}_0 \mathcal{Q}_{\mathcal{I}_2} \mid {\sf W}_{[K]\setminus \{k\}} \big)}_{=0}   + I\big({\sf W}_{k};  \mathcal{Q}_{\mathcal{J}}  \mid \mathcal{Q}_0 \mathcal{Q}_{\mathcal{I}_2}, {\sf W}_{[K]\setminus \{k\}} \big)  \\
			&= I\big({\sf W}_{k};  \mathcal{Q}_{\mathcal{J}}  \mid \mathcal{Q}_0 \mathcal{Q}_{\mathcal{I}_2}, {\sf W}_{[K]\setminus \{k\}} \big) \label{eq:conv_2_2} \\
			&\leq H(\mathcal{Q}_{\mathcal{J}}) - H(\mathcal{Q}_{\mathcal{J}} \mid \mathcal{Q}_0 \mathcal{Q}_{\mathcal{J}_2} , {\sf W}_{[K]}) \label{eq:conv_2_3} \\
			&\leq (\alpha-\beta) \log_d \delta - H(\mathcal{Q}_{\mathcal{J}} \mid \mathcal{Q}_0 \mathcal{Q}_{\mathcal{J}_2} , {\sf W}_{[K]}) \label{eq:conv_2_4}
		\end{align}
		Information measures on and after Step \eqref{eq:conv_2_1} are with respect to the state $\rho_{{\sf W}_1\cdots {\sf W}_K \mathcal{Q}_0 \mathcal{Q}_1\cdots \mathcal{Q}_S}$. Step \eqref{eq:conv_2_1} is by Lemma \ref{lem:Holevo} since Alice obtains $Y_t$ conditioned on any realization of ${\sf W}_{[K]\setminus \{k\}}$ by measuring $\mathcal{Q}_0\mathcal{Q}_{\mathcal{J}}\mathcal{Q}_{\mathcal{I}_2}$. Step \eqref{eq:conv_2_2} is by Lemma \ref{lem:no_signal} conditioned on any realization of ${\sf W}_{[K] \setminus \{k\}}$, since ${\sf W}_k$ is not available to any server in $\{0\} \cup \mathcal{I}_2$. Step \eqref{eq:conv_2_3} is because conditioning does not increase entropy.
		 
		By symmetry, \eqref{eq:conv_2_4} holds for any $\mathcal{J} \in \binom{[2\alpha-S]}{\alpha-\beta}$. Taking the sum over all $\mathcal{J} \in \binom{[2\alpha-S]}{\alpha-\beta}$ (in total $\binom{2\alpha-S}{\alpha-\beta}$ terms), we have
		\begin{align}
			& \binom{2\alpha-S}{\alpha-\beta}L \leq \binom{2\alpha-S}{\alpha-\beta}(\alpha-\beta) \log_d \delta   - \sum_{\mathcal{J}} H(\mathcal{Q}_\mathcal{J} \mid \mathcal{Q}_0 \mathcal{Q}_{\mathcal{I}_2}, {\sf W}_{[K]}) \\
			& \leq \binom{2\alpha-S}{\alpha-\beta} (\alpha-\beta) \log_d \delta  - \binom{2\alpha-S-1}{\alpha-\beta-1} H(\mathcal{Q}_{[2\alpha-S]} \mid \mathcal{Q}_0 \mathcal{Q}_{\mathcal{I}_2}, {\sf W}_{[K]}) \label{eq:conv_2_5}
		\end{align}
		\begin{align}
			\implies L &\leq (\alpha-\beta) \log_d \delta  - \frac{\alpha-\beta}{2\alpha-S} H(\mathcal{Q}_{[2\alpha-S]} \mid \mathcal{Q}_0 \mathcal{Q}_{\mathcal{I}_2}, {\sf W}_{[K]}) \\
			& = (\alpha-\beta) \log_d \delta - \frac{\alpha-\beta}{2\alpha-S} H(\mathcal{Q}_{[2\alpha-S]} \mathcal{Q}_0 \mid  \mathcal{Q}_{\mathcal{I}_2}, {\sf W}_{[K]}) + \frac{\alpha-\beta}{2\alpha-S}H(\mathcal{Q}_0 \mid \mathcal{Q}_{\mathcal{I}_2}, {\sf W}_{[K]}) \label{eq:conv_2_6}
		\end{align}
		Step \eqref{eq:conv_2_5} follows from \cite[Lemma 3]{grassl2022entropic} (which may be viewed as the quantum conditional version of Han's inequality).
		By similar reasoning, we have ($\mathcal{I}_2$ replaced by $\mathcal{I}_3$),
		\begin{align}
			L &\leq (\alpha-\beta) \log_d \delta  - \frac{\alpha-\beta}{2\alpha-S} H(\mathcal{Q}_{[2\alpha-S]} \mathcal{Q}_0 \mid  \mathcal{Q}_{\mathcal{I}_3}, {\sf W}_{[K]}) + \frac{\alpha-\beta}{2\alpha-S}H(\mathcal{Q}_0 \mid \mathcal{Q}_{\mathcal{I}_3}, {\sf W}_{[K]}). \label{eq:conv_2_7}
		\end{align}
		Adding \eqref{eq:conv_2_6} and \eqref{eq:conv_2_7}, and noting that,
		\begin{align}
			&H(\mathcal{Q}_{[2\alpha-S]} \mathcal{Q}_0 \mid  \mathcal{Q}_{\mathcal{I}_2}, {\sf W}_{[K]})  + H(\mathcal{Q}_{[2\alpha-S]} \mathcal{Q}_0 \mid  \mathcal{Q}_{\mathcal{I}_3}, {\sf W}_{[K]}) \geq 0,
		\end{align}
		because of weak monotonicity conditioned on ${\sf W}_{[K]} = w$ for any $w\in \mathbb{F}_d^{KL}$, we have
		\begin{align}
			2L &\leq 2(\alpha-\beta) \log_d \delta + \frac{2(\alpha-\beta)}{2\alpha-S} \log_d \delta_0\\
			& \hspace{-1.2cm} \implies  1 \leq (\alpha-\beta)\Delta + \frac{\alpha-\beta}{2\alpha-S}\Delta_0 \label{eq:conv_2_8} \\
			& \hspace{-1.2cm} \implies  \Delta \geq \frac{1}{\alpha-\beta} - \frac{\Delta_0}{2\alpha-S} 
		\end{align}
		In Step \eqref{eq:conv_2_8} we divide both sides by $2L$  and apply the definitions $\Delta_0 = \log_d \delta_0/L, \Delta = \log_d \delta/L$.
\end{enumerate}
This concludes the proof of the converse. \hfill \qed

\section{Conclusion}\label{sec:conc}
Generalizing the $\Sigma$-QMAC model to $\Sigma$-QEMAC, which allows erasures, we construct a coding scheme to show that the communication efficiency gains that are enabled by quantum multiparty entanglement are not dominated in general by the higher cost of quantum erasure correction. 
The advantage provided by quantum entanglement is evident, e.g., from Corollaries \ref{cor:no_helper} and \ref{cor:sym_largeE}.  For the symmetric setting we also show by quantum entropic analysis that the proposed coding scheme, in combination with the simple idea of treating qudits as classical dits, is exactly optimal in communication efficiency for the $\Sigma^s$-QEMAC, for any given level of prior entanglement between the data-servers and the receiver. In other words, we characterize the precise capacity of the $\Sigma^s$-QEMAC for arbitrary levels of receiver entanglement. A natural next step would be to find the capacity of the asymmetric $\Sigma$-QEMAC.  Additional insights are needed to overcome the combinatorial complexity of asymmetric settings. This is left as an open problem for future work.

\appendix
\section{Proof of achievability of $\Delta = 3/4$} \label{proof:communication}
Let us first formulate the problem. The reduced communication problem contains $3$ servers, referred to as Server $1$, Server $2$ and Server $3$. The answer from one of them can get erased. The data stream ${\sf A}$ is only known to Server $1$ and Server $2$. The receiver must be able to recover ${\sf A}$.
Now, let us prove that $\Delta=3/4$ is achievable for the communication problem reduced from the ``${\sf A}$ cut" from Section \ref{sec:example}. Let $[{\sf A}_1, {\sf A}_2, {\sf A}_3, {\sf A}_4]$ be symbols from $\mathbb{F}_{d^z}$ that constitute $4z$ instances of the data stream ${\sf A}$. Let 
\begin{align}
	\begin{bmatrix}
		{\sf A}_1' & {\sf A}_2' & \cdots & {\sf A}_8'
	\end{bmatrix}
	=
	\begin{bmatrix}
		{\sf A}_1 & {\sf A}_2 & {\sf A}_3 & {\sf A}_4
	\end{bmatrix}
	{\bf M}^{4\times 8}
\end{align}
where ${\bf M}\in \mathbb{F}_{d^z}^{4\times 8}$ is the generator matrix of a $(k=4,n=8)$ MDS linear code defined in $\mathbb{F}_{d^z}$. $z$ is allowed to be chosen freely to allow the existence of ${\bf M}$. Let $Q_1,Q_2,\cdots, Q_8$ be $d^z$-dimensional quantum systems, distributed to Server $1$, Server $2$ and Server $3$ such that
\begin{enumerate}
	\item Server $1$ has $Q_1$, $Q_2$, $Q_3$;
	\item Server $2$ has $Q_4$, $Q_5$, $Q_6$;
	\item Server $3$ has $Q_7$, $Q_8$.
\end{enumerate}
Let $(Q_1, Q_7)$ and $(Q_4,Q_8)$ be initially prepared in the entangled state for superdense coding. Recall that only Server $1$ and Server $2$ know $({\sf A}_1',\cdots,{\sf A}_8')$. Server $1$ and Server $2$ perform the encoding operations as follows.
\begin{enumerate}
	\item Server $1$ encodes $({\sf A}_1', {\sf A}_2')$ into $Q_1$ using superdense coding, and encodes ${\sf A}_3'$ into $Q_2$, ${\sf A}_4'$ into $Q_3$ classically.
	\item Server $2$ encodes $({\sf A}_5', {\sf A}_6')$ into $Q_4$ using superdense coding, and encodes ${\sf A}_7'$ into $Q_5$, ${\sf A}_8'$ into $Q_6$ classically.
\end{enumerate}
Depending on different cases of erasure, the receiver applies one of the following decoding options.
\begin{enumerate}
	\item In the case that the answer from Server $1$ gets erased (this also ruins $Q_7$), Alice is able to measure $Q_4, Q_5, Q_6, Q_8$ to obtain $({\sf A}_5', {\sf A}_6', {\sf A}_7', {\sf A}_8')$ and recover $({\sf A}_1, {\sf A}_2, {\sf A}_3, {\sf A}_4)$ by the property of the MDS code.
	\item In the case that the answer from Server $2$ gets erased (this also ruins $Q_8$), Alice is able to measure $Q_1, Q_2, Q_3, Q_7$ to obtain $({\sf A}_1', {\sf A}_2', {\sf A}_3', {\sf A}_4')$ and recover $({\sf A}_1, {\sf A}_2, {\sf A}_3, {\sf A}_4)$ by the property of the MDS code..
	\item In the case that the answer from Server $3$ gets erased (this also ruins $Q_1, Q_4$), Alice is able to measure $(Q_2, Q_3, Q_5, Q_6)$ to obtain $({\sf A}_3', {\sf A}_4', {\sf A}_7', {\sf A}_8')$ and recover $({\sf A}_1, {\sf A}_2, {\sf A}_3, {\sf A}_4)$ by the property of the MDS code..
\end{enumerate}
Therefore, Alice gets $4z$ instances of the data stream ${\sf A}$ by downloading at most $3z$ qudits from each of the servers. We thus conclude that $\Delta = 3/4$ is achievable for the communication problem.

\section{Proof of Lemma \ref{lem:box}} \label{proof:box}
Recall that we are given $\mathbb{F}_q, N$ and an SSO matrix ${\bf M} = [{\bf M}_l, {\bf M}_r]$. Let
\begin{align}
	{\bf J} \triangleq \begin{bmatrix}
		{\bf 0} & -{\bf I} \\ {\bf I} & {\bf 0}
	\end{bmatrix}
\end{align}
where ${\bf 0}$ is the $N\times N$ zero matrix and ${\bf I}$ is the $N\times N$ identity matrix. In this section, when we say an $N\times 2N$ matrix ${\bf M}$ is SSO, this means that ${\bf M}$ has full row rank equal to $N$ and that ${\bf M}{\bf J}{\bf M}^\top = {\bf 0}$; when we say an  $2N\times N$ matrix ${\bf G}$ is SSO, this means that ${\bf G}$ has full column rank equal to $N$ and that ${\bf G}^\top {\bf J}{\bf G} = {\bf 0}$.

Define a $2N\times N$ $\mathbb{F}_q$ matrix
\begin{align}
	{\bf G} \triangleq \begin{bmatrix}
		{\bf M}_r^\top \\ -{\bf M}_l^\top
	\end{bmatrix} 
	\triangleq
	\begin{bmatrix}
		{\bf A}  \\ {\bf B}  
	\end{bmatrix}.
\end{align}
Note that ${\bf G}^\top{\bf J}{\bf G} = {\bf M}_r{\bf M}_l^\top - {\bf M}_l {\bf M}_r^\top = {\bf 0}$, which shows that ${\bf G}$ is SSO. It is known \cite[Remark 1]{Allaix_N_sum_box} that an SSO matrix can be completed to a symplectic basis by a skew-symmetric version of a Gram-Schmidt process  \cite[Theorem 1.1]{Ana_Sym}. In other words, there exists a matrix ${\bf H} = \bbsmatrix{{\bf C} \\ {\bf D}} \in \mathbb{F}_q^{2N\times N}$ such that $[{\bf G},  {\bf H}] \triangleq {\bf F}$ is \emph{symplectic}, i.e., ${\bf F}^\top {\bf J}{\bf F} = {\bf J}$. It can be verified that  $\bbsmatrix{{\bf D}^{\top} & -{\bf C}^{\top} \\ -{\bf B}^{\top} & {\bf A}^{\top}}{\bf F}$ is the $2N\times 2N$ identity matrix, so ${\bf F}^{-1}  = \bbsmatrix{{\bf D}^{\top} & -{\bf C}^{\top} \\ -{\bf B}^{\top} & {\bf A}^{\top}}$. With the ${\bf G}$ and ${\bf H}$ so defined, we have, 
\begin{align} \label{eq:GHFM}
	[{\bf 0}, {\bf I}][{\bf G}, {\bf H}]^{-1} = [{\bf 0}, {\bf I}]{\bf F}^{-1} = [-{\bf B}^\top , {\bf A}^\top ] = [{\bf M}_l, {\bf M}_r] = {\bf M}.
\end{align}

To proceed, we need the following definitions from the quantum stabilizer formalism over a finite field (\!\!\cite[Sec. IV-A]{song_colluding_PIR}, \cite[Sec. II]{Allaix_N_sum_box}). 
\begin{enumerate}
	\item $Q_1,Q_2,\cdots, Q_N$ denote $N$ $q$-dimensional quantum systems, which together constitute a $q^N$-dimensional quantum system. The Hilbert space corresponding to $Q_i$ is denoted as $\mathcal{H}_i$, and the Hilbert space for the joint quantum system is denoted as $\mathcal{H}_1\otimes \cdots \otimes \mathcal{H}_N$.
	\item For ${\bf x} = [x_1,\cdots, x_N]^\top  \in \mathbb{F}_q^{N\times 1}, {\bf z}=[z_1,\cdots, z_N]^\top  \in \mathbb{F}_q^{N\times 1}$, ${\bf W}(\bbsmatrix{{\bf x}\\ {\bf z}})$ is the Weyl operator defined as $c_{{\bf x}, {\bf z}}{\sf X}(x_1){\sf Z}(z_1)\otimes \cdots \otimes {\sf X}(x_N){\sf Z}(z_N)$. For our purpose the global phase $c_{{\bf x}, {\bf z}}$ does not matter, so we suppress $c_{{\bf x}, {\bf z}}$.
	\item Given a linear subspace $\mathcal{V} \subseteq \mathbb{F}_q^{2N\times 1}$, define $\mathcal{V}^{\bot_J} \triangleq \{{\bf w} \in \mathbb{F}_q^{2N\times 1} \colon \tr({\bf v}^\top {\bf J} {\bf w}) = 0,~ \forall {\bf v} \in \mathcal{V}\}$.  A linear subspace $\mathcal{V} \subseteq \mathbb{F}_q^{2N\times 1}$ is called self-orthogonal if and only if $\mathcal{V} \subseteq \mathcal{V}^{\bot_J}$.
	\item Given a self-orthogonal linear subspace $\mathcal{V}\subseteq \mathbb{F}_q^{2N\times 1}$, let $\mathcal{S}(\mathcal{V})$ be a stabilizer group defined from $\mathcal{V}$ \cite[Prop, 1]{song_colluding_PIR}. 
	\item Given a linear subspace $\mathcal{V}\subseteq \mathbb{F}_q^{2N\times 1}$, for a vector ${\bf s} \in \mathbb{F}_q^{2N\times 1}$, define the coset $\bar{\bf s} \triangleq \{{\bf s} + {\bf v} \colon {\bf v} \in \mathcal{V}\}$. All such cosets constitute a quotient space denoted as $\mathbb{F}_q^{2N\times 1}/\mathcal{V}$.
\end{enumerate}

Now let us recall \cite[Prop. 2]{song_colluding_PIR}, (see also \cite[Prop. 1]{Allaix_N_sum_box}, \cite[Prop. 2.2]{QMDSTPIR}).
\begin{proposition}[\!\!\!{\cite[Prop. 2]{song_colluding_PIR}}] \label{prop:stabilizer}
	Let $\mathcal{V}$ be a $D$-dimensional self-orthogonal subspace of $\mathbb{F}_q^{2N\times 1}$ and $\mathcal{S}(\mathcal{V})$ be a stabilizer defined from $\mathcal{V}$. For a coset $\bar{\bf s}\in \mathbb{F}_q^{2N\times 1}/\mathcal{V}^{\bot_J}$, let ${\bf s}$ be its coset leader. Then the following statements hold.
	\begin{enumerate}[label=(\alph*)]
		\item For any ${\bf v} \in \mathcal{V}$, the operation ${\bf W}({\bf v}) \in \mathcal{S}(\mathcal{V})$ is simultaneously and uniquely decomposed as 
		\begin{align}
			{\bf W}({\bf v}) = \sum_{\bar{\bf s} \in \mathbb{F}_q^{2N \times 1}/\mathcal{V}^{\bot_J}} \omega^{\tr({\bf v}^\top {\bf J} {\bf s})} {\bf P}_{\bar{\bf s}}
		\end{align}
		with orthogonal projections $\{{\bf P}_{\bar{\bf s}}\}$ such that 
		\begin{align}
			&{\bf P}_{\bar{\bf s}}{\bf P}_{\bar{\bf t}} = {\bf 0}~ \mbox{for any}~ \bar{\bf s} \not= \bar{\bf t} \notag,\\
			&\sum_{\bar{\bf s} \in \mathbb{F}_q^{2N \times 1}/\mathcal{V}^{\bot_J}} {\bf P}_{\bar{\bf s}} = {\bf I}_{q^N}. \notag
		\end{align}
 		\item Let $\Image{\bf P}_{\bar{\bf s}}$ be the image of the projection ${\bf P}_{\bar{\bf s}}$. For any ${\bf s},{\bf t} \in \mathbb{F}_d^{2N\times 1}$,
 		\begin{align}
 			{\bf W}({\bf s})(\Image{\bf P}_{\bar{\bf s}}) = \Image{\bf P}_{\overline{{\bf s}+ {\bf t}}}.
 		\end{align}
 		\item For any $\bar{\bf s} \in \mathbb{F}_q^{2N\times 1}/\mathcal{V}^{\bot_J}$,
 		\begin{align}
 			\dim (\Image{\bf P}_{\bar{\bf s}}) = q^{N-D}.
 		\end{align}
	\end{enumerate}
\end{proposition}
Let us apply  Proposition \ref{prop:stabilizer} to the special case of $\mathcal{V} = \langle {\bf G} \rangle$, where $\langle {\bf G} \rangle$ is the column-span of ${\bf G}$ with coefficients in $\mathbb{F}_q$. First note that since ${\bf G}$ is SSO, we have $\mathcal{V} = \mathcal{V}^{\bot_J} = \langle {\bf G} \rangle$ \cite[Rem. 4]{Allaix_N_sum_box}, which yields $\mathcal{S}(\mathcal{V})$ a maximal stabilizer (group). Then according to (c), each $\Image{\bf P}_{\bar{\bf s}}$ is a $1$-dimensional subspace, and we denote $\ket{\bar{\bf s}}_*$ as a unit vector from $\Image{\bf P}_{\bar{\bf s}}$ for each coset $\bar{\bf s}\in \mathbb{F}_q^{2N\times 1}/\langle {\bf G} \rangle$. According to (a), we have that $\{\ket{\bar{\bf s}}\colon \bar{\bf s}\in \mathbb{F}_q^{2N\times 1}/\langle {\bf G} \rangle\}$ is a set of orthogonal vectors, and since the cardinality of the set is $q^N$, they form a basis of $\mathcal{H}_1 \otimes \cdots \otimes \mathcal{H}_N$. According to (b), we have for  ${\bf s}, {\bf t} \in \mathbb{F}_q^{2N\times 1}$, 
\begin{align}\label{eq:stabilizer_state_evolve}
	{\bf W}({\bf t}) \ket{\bar{\bf s}}_{*} = \ket{\overline{ {\bf s}+{\bf t} }}_{*}.
\end{align}

For a vector ${\bf s}\in \mathbb{F}_q^{2N\times 1}$, define $g({\bf s}) \in \mathbb{F}_q^{N\times 1}$ and $h({\bf s}) \in \mathbb{F}_q^{N\times 1}$ as the unique coefficient vectors that yield ${\bf s} = {\bf G}g({\bf s})  + {\bf H}h({\bf s})$. This is feasible because $[{\bf G}, {\bf H}]$ is invertible. Note that from this definition we have $h({\bf s}) = [{\bf 0}, {\bf I}][{\bf G}, {\bf H}]^{-1}{\bf s} = {\bf M}{\bf s}$ by \eqref{eq:GHFM}.
For a vector ${\bf s} \in \mathbb{F}_q^{2N\times 1}$, define the coset $\bar{\bf s}\triangleq \{{\bf s} + {\bf g} \colon {\bf g} \in \langle {\bf G} \rangle\}\in \mathbb{F}_q^{2N\times 1}/\langle {\bf G} \rangle$. It follows that given ${\bf G}, {\bf H}$, there is a one-to-one correspondence between a coset $\bar{\bf s}$ and $h({\bf s}) = {\bf M}{\bf s}$. This is because for ${\bf s}, {\bf t} \in \mathbb{F}_q^{2N\times 1}$, $\bar{\bf s} = \bar{\bf t}$ if and only if $h({\bf s}) = h({\bf t})$. Denote this bijection as $f\colon \mathbb{F}_q^{2N\times 1}/\langle {\bf G}\rangle \to \mathbb{F}_q^{N\times 1}$.
Now, due to the bijection $f$, let us define $\ket{{\bf M}{\bf s}}_{\bf M} \triangleq \ket{\bar{\bf s}}_*$ for ${\bf s} \in \mathbb{F}_q^{2N\times 1}$. It follows that $\{\ket{{\bf M}{\bf s}}_{\bf M} \colon {\bf s} \in \mathbb{F}_q^{2N\times 1}\} \stackrel{(*)}{=} \{\ket{{\bf a}}_{\bf M} \colon {\bf a} \in \mathbb{F}_q^{N\times 1}\}$ is a set of orthogonal states that constitute a basis for $\mathcal{H}_1\otimes \cdots \otimes \mathcal{H}_N$. Step (*) follows as ${\bf M}$ has full row rank equal to $N$. Then for any ${\bf a}, {\bf x}, {\bf z} \in \mathbb{F}_q^{N\times 1}$, if the initial state of $Q_1\cdots Q_N$ is $\ket{{\bf a}}_{\bf M} = \ket{{\bf M}{\bf s}}_{\bf M}$ for some ${\bf s}\in \mathbb{F}_q^{2N\times 1}$, with ${\bf t}$ set to $\bbsmatrix{{\bf x} \\{\bf z}}$, \eqref{eq:stabilizer_state_evolve} implies that applying ${\sf X}(x_i){\sf Z}(z_i)$ to $Q_i$ for $i\in [N]$, the state evolves to
\begin{align}
	\ket{\overline{ {\bf s}+{\bf t} }}_{*} = \ket{ {\bf M} ({\bf s} + \bbsmatrix{{\bf x}\\{\bf z}}) }_{{\bf M}} = \ket{{\bf a} + {\bf M} \bbsmatrix{{\bf x}\\{\bf z}}}_{{\bf M}}.
\end{align}
This recovers Lemma \ref{lem:box}. Let us note that essentially  what this proof is doing is ``relabeling" the set of orthogonal states $\{\ket{\bar{\bf s}}_{*}\}$ with $\{\ket{{\bf M}{\bf s}}_{{\bf M}}\}$ by identifying the bijection $f$.

\section{Proof of Lemma \ref{lem:independent_mixed}} \label{proof:independent_mixed}
Let $\mathcal{H}_A, \mathcal{H}_B$ denote the Hilbert spaces related to System $A$ and System $B$. It suffices to prove for pure state $\rho_{AB} = \ket{\psi}_{AB}\bra{\psi}_{AB}$.  Let us first show the case for $d = p$ being a prime. Without loss of generality, we write
\begin{align} \label{eq:general_joint_state}
	\ket{\psi}_{AB} = \sum_{i=0}^{p-1}\ket{\phi_i}_{A}\ket{i}_B
\end{align}
for a set of vectors $\{\ket{\phi_i}_{A}\}_{i=1}^{p-1}$ defined in $\mathcal{H}_A$, and $\{\ket{i}_{B}\}_{i=0}^{p-1}$ is the computational basis for $\mathcal{H}_B$.

After the random ${\sf X}(\tilde{x}){\sf Z}(\tilde{z})$ operation is applied to System $B$, the state becomes a mixed state with density operator
\begin{align} \label{eq:post_state_1}
	\rho_{AB}' &= \frac{1}{p^2}\sum_{\tilde{x}=0}^{p-1} \sum_{\tilde{z}=0}^{p-1}  \underbrace{\big( {\bf I}\otimes {\sf X}(\tilde{x}){\sf Z}(\tilde{z}) \big)}_{U_{\tilde{x}\tilde{z}}}\ket{\psi}_{AB} \bra{\psi}_{AB}U_{\tilde{x}\tilde{z}}^\dagger.
\end{align}
By \eqref{eq:general_joint_state},
\begin{align}
	& U_{\tilde{x}\tilde{z}} \ket{\psi}_{AB}\bra{\psi}_{AB}U_{\tilde{x}\tilde{z}}^\dagger \notag\\
	& = \sum_{i = 0}^{p-1}\sum_{j=0}^{p-1} \omega^{\tilde{z}(i-j)} \ket{\phi_i}_A\bra{\phi_j}_A\ket{i+\tilde{x}}_B\bra{j+\tilde{x}}_B.
\end{align}
Therefore, \eqref{eq:post_state_1} is equal to (the subscripts $A,B$ are omitted for simplicity)
\begin{align} \label{eq:post_state_2}
	& \frac{1}{p^2}\sum_{\tilde{x}=0}^{p-1} \sum_{\tilde{z}=0}^{p-1}  \sum_{i = 0}^{p-1}\sum_{j=0}^{p-1} \omega^{\tilde{z}(i-j)} \ket{\phi_i}\bra{\phi_j}\ket{i+\tilde{x}}\bra{j+\tilde{x}} \\
	&= \frac{1}{p^2} \sum_{i = 0}^{p-1}\sum_{j=0}^{p-1} \ket{\phi_i}\bra{\phi_j} \sum_{\tilde{x}=0}^{p-1} \sum_{\tilde{z}=0}^{p-1}  \omega^{\tilde{z}(i-j)} \ket{i+\tilde{x}}\bra{j+\tilde{x}} \\
	&= \frac{1}{p^2} \sum_{i = 0}^{p-1}\sum_{j=0}^{p-1} \ket{\phi_i}\bra{\phi_j}  \underbrace{\sum_{\tilde{z}=0}^{p-1}  \omega^{\tilde{z}(i-j)}}_{p\delta(i-j)} \sum_{\tilde{x}=0}^{p-1} \ket{i + \tilde{x}}\bra{j +  \tilde{x}} \\
	& =\frac{1}{p} \sum_{i = 0}^{p-1}\sum_{j=0}^{p-1} \delta(i-j) \ket{\phi_i}\bra{\phi_j}      \sum_{\tilde{x}=0}^{p-1} \ket{i + \tilde{x}}\bra{j + \tilde{x}}\\
	& =\frac{1}{p} \sum_{i = 0}^{p-1} \ket{\phi_i}\bra{\phi_i} \underbrace{\sum_{\tilde{x}=0}^{p-1} \ket{i +\tilde{x}}\bra{i + \tilde{x}}}_{{\bf I}_p}\\
	&= \Big(\sum_{i=0}^{p-1} \ket{\phi_i}\bra{\phi_i}\Big) \otimes \Big( \frac{{\bf I}_p}{p} \Big)
\end{align}
where $\delta(x) \triangleq 1$ if $x=0$ and $\delta(x) \triangleq 0$ otherwise.

This shows that $B$ is independent of $A$ and is in the maximally mixed state. To generalize the argument to general $d = p^r$, it is important to note that applying a random ${\sf X}(\tilde{x}){\sf Z}(\tilde{z})$ operation to a qudit is equivalent to applying a random ${\sf X}(\tilde{x}_i){\sf Z}(\tilde{x}_i)$ operation to the $i^{th}$ $p$-dimensional subsystem for all $i\in [r]$, where $\tilde{x}_i, \tilde{z}_i \in \mathbb{F}_p$. Therefore, the argument generalizes to $d=p^r$.

\section{Proofs of claims in Section \ref{proof:achievability}} \label{proof:schwartz_zippel}
\subsection{The existence of ${\bf U}$}
Recall that we are given ${\bf E}_t \in \mathbb{F}_q^{N\times \sum_{s\in \mathcal{E}(t)}2N_s}$ that has full column rank for $t\in [T]$. We will prove the existence of the matrix ${\bf U} \in \mathbb{F}_q^{N\times u}$, where $u = N-\max_{t\in [T]}\sum_{s\in \mathcal{E}(t)}N_s$, such that $[{\bf U}, {\bf E}_t]$ has full column rank for all $t\in [T]$.
For $t\in [T]$, there exists a realization of ${\bf U} \in \mathbb{F}_q^{N\times u}$ and a deterministic matrix ${\bf Z}_t \in \mathbb{F}_q^{N\times (N-u-\sum_{s\in \mathcal{E}(t)}2N_s)}$ such that $[{\bf U}, {\bf E}_t, {\bf Z}_t]$ is invertible. Therefore, $P_t \triangleq \det([{\bf U}, {\bf E}_t, {\bf Z}_t])$ is a non-zero polynomial in the elements of ${\bf U}$, with degree more than $N$ for all $t\in [T]$. Consider the polynomial $P = \prod_{t=1}^T P_t$, which is a non-zero polynomial with degree not more than $T N$. By Schwartz-Zippel Lemma, if the element of ${\bf U}$ is chosen i.i.d. uniformly in $\mathbb{F}_q$, then the probability of $P$ evaluating to $0$ is not more than $\frac{T N}{q}$, which is strictly less than $1$ if $q>T N$. Therefore, if $q>T N$, there exists ${\bf U}$ such that $[{\bf U}, {\bf E}_t]$ has full column rank for all $t\in [T]$.

\subsection{The existence of ${\bf V}_{\dec}$}
Recall that we are given ${\bf U}_k \in \mathbb{F}_q^{N\times \scalebox{0.8}{\rk}({\bf U}_k)}$, ${\bf V}_k' \in \mathbb{F}_q^{\scalebox{0.8}{\rk}({\bf U}_k) \times l}$ that has full column rank for $k\in [K]$. We will prove the existence of the matrix ${\bf V}_{\dec} \in \mathbb{F}_q^{l\times N}$ such that ${\bf V}_{\dec}{\bf U}_k {\bf V}'_k$ is an $l\times l$ invertible matrix ${\bf R}_k$ for all $k \in [K]$. The proof is as follows. By definition, there exists a realization of ${\bf V}_{\dec} \in \mathbb{F}_q^{l \times N}$ such that ${\bf V}_{\dec}{\bf U}_k{\bf V}'_k$ is invertible. Therefore, $P_k \triangleq \det({\bf V}_{\dec}{\bf U}_k{\bf V}'_k)$ is a non-zero polynomial in the elements of ${\bf V}_{\dec}$, with degree not more than $l$ for all $k\in [K]$. Consider the polynomial $P = \prod_{k=1}^K P_K$, which is a non-zero polynomial with degree not more than $Kl$. By Schwartz-Zippel Lemma, if the element of ${\bf V}_{\dec}$ is chosen i.i.d. uniformly in $\mathbb{F}_q$, then the probability of $P$ evaluating to $0$ is not more than $\frac{Kl}{q}$, which is strictly less than $1$ if $q>Kl$. Therefore, if $q>Kl$, there exists ${\bf V}_{\dec}$ such that ${\bf V}_{\dec}{\bf U}_k {\bf V}'_k$ is invertible for all $k\in [K]$.

\section{Proof of Theorem \ref{thm:eacq_bound}} \label{proof:eacq_bound}
We continue to use the notations as  defined in Section \ref{proof:thm_conv}.
After the coding operations done by the servers, the classical-quantum system of interest  ${\sf W}\mathcal{Q}_0\mathcal{Q}_1 \cdots\mathcal{Q}_S$ is in the state,
\begin{align}
	\rho_{{\sf W} \mathcal{Q}_0\mathcal{Q}_1\cdots \mathcal{Q}_S} = \sum_{w\in  \mathbb{F}_d^{L}}\frac{1}{d^{L}}\rho_{\mathcal{Q}_0\mathcal{Q}_1\cdots \mathcal{Q}_S}^{(w)}.
\end{align}
For $t\in [T]$, Alice measures $\mathcal{Q}_0\mathcal{Q}_{[S]\setminus \mathcal{E}(t)}$ to obtain $Y_t$ such that ${\sf Pr}(Y_t = {\sf W}) = 1$. Recall that in this theorem $\{\mathcal{E}(1), \mathcal{E}(2), \cdots, \mathcal{E}(T)\}$ constitute the collection of all cardinality-$\beta$ subsets of $\{1,2,\cdots, S\}$. Denote $\mathcal{I}_t \triangleq [S]\setminus \mathcal{E}(t)$. We have, for all $t\in [T]$, 
\begin{align}
	L &= H({\sf W}) = I({\sf W}; Y_t) \\
	&\leq I({\sf W}; \mathcal{Q}_0 \mathcal{Q}_{\mathcal{I}_t})_{\rho_{{\sf W} \mathcal{Q}_0\mathcal{Q}_1\cdots\mathcal{Q}_S}} \label{eq:eacq_1_1} \\
	&= \underbrace{I({\sf W}; \mathcal{Q}_0)}_{=0} + I({\sf W}; \mathcal{Q}_{\mathcal{I}_t} \mid \mathcal{Q}_0) \label{eq:eacq_1_2} \\
	&= I({\sf W}; \mathcal{Q}_{\mathcal{I}_t} \mid \mathcal{Q}_0) \\
	&= H(\mathcal{Q}_{\mathcal{I}_t} \mid \mathcal{Q}_0) - H(\mathcal{Q}_{\mathcal{I}_t}\mid \mathcal{Q}_0, {\sf W}) \\
	&\leq H(\mathcal{Q}_{\mathcal{I}_t}) - H(\mathcal{Q}_{\mathcal{I}_t}\mid \mathcal{Q}_0, {\sf W})  \label{eq:eacq_1_3} \\
	&\leq H(\mathcal{Q}_{\mathcal{I}_t}) + H(\mathcal{Q}_{\mathcal{I}_t}) \label{eq:eacq_1_4} \\
	&\leq 2\log_d \sum_{s\in \mathcal{I}_t}\delta_s \\
	& \hspace{-1cm} \implies \sum_{s\in \mathcal{I}} \Delta_s \geq 1/2, ~~\forall \mathcal{I} \in \binom{[S]}{S-\beta}
\end{align}
Information measures on and after Step \eqref{eq:eacq_1_1} are with respect to the state $\rho_{{\sf W} \mathcal{Q}_0 \mathcal{Q}_1\cdots \mathcal{Q}_S}$. Step \eqref{eq:eacq_1_1} is by Holevo bound (Lemma \ref{lem:Holevo}). Step \eqref{eq:eacq_1_2} is by Lemma \ref{lem:no_signal} since Server $0$ does not know ${\sf W}$. Step \eqref{eq:eacq_1_3} is because conditioning does not increase entropy. \eqref{eq:eacq_1_4} is from the Ariki-Lieb triangle inequality\footnote{For a bipartite quantum system $AB$ in the state $\rho$, $|H(A)_{\rho}-H(B)_{\rho}| \leq H(AB)_{\rho}$.}.
This proves \eqref{eq:eacq_1}. We next show \eqref{eq:eacq_2}. From \eqref{eq:eacq_1_3}, we have,
\begin{align}
	L \leq H(\mathcal{Q}_{\mathcal{I}}) - H(\mathcal{Q}_{\mathcal{I}} \mid \mathcal{Q}_0, {\sf W}),~~\forall \mathcal{I} \in \binom{[S]}{S-\beta}.
\end{align}
Taking the sum over all $\mathcal{I} \in \binom{[S]}{S-\beta}$ (in total $\binom{S}{\beta}$ terms), we have
\begin{align}
	&\binom{S}{\beta}L \leq \sum_{\mathcal{I} \in  \binom{[S]}{S-\beta}}   H(\mathcal{Q}_{\mathcal{I}}) - \sum_{\mathcal{I} \in  \binom{[S]}{S-\beta}} H(\mathcal{Q}_{\mathcal{I}} \mid \mathcal{Q}_0, {\sf W})   \\
	&\leq \binom{S-1}{\beta} \sum_{s\in [S]}\log_d \delta_s - \sum_{\mathcal{I} \in  \binom{[S]}{S-\beta}} H(\mathcal{Q}_{\mathcal{I}} \mid \mathcal{Q}_0, {\sf W})  \\
	&\leq \binom{S-1}{\beta} \sum_{s\in [S]}\log_d \delta_s -  \binom{S-1}{\beta} H(\mathcal{Q}_{[S]} \mid \mathcal{Q}_0, {\sf W}) \label{eq:eacq_2_1}\\
	&\leq \binom{S-1}{\beta} \sum_{s\in [S]}\log_d \delta_s +  \binom{S-1}{\beta} H(\mathcal{Q}_0) \label{eq:eacq_2_2}\\
	&\leq \binom{S-1}{\beta} \big(\sum_{s\in [S]}\Delta_s + \Delta_0\big)L \\
	&   \implies \Delta_0+\Delta_1+\cdots+\Delta_S \geq \frac{S}{S-\beta}
\end{align}
Step \eqref{eq:eacq_2_1} follows from \cite[Lemma 3]{grassl2022entropic}. Step \eqref{eq:eacq_2_2} is due to non-negativity of entropy and the fact that conditioning does not increase entropy. This concludes the proof of Theorem \ref{thm:eacq_bound}.

\bibliographystyle{IEEEtran}
\bibliography{../../bib_file/yy.bib}

\begin{thebibliography}{10}
\providecommand{\url}[1]{#1}
\csname url@samestyle\endcsname
\providecommand{\newblock}{\relax}
\providecommand{\bibinfo}[2]{#2}
\providecommand{\BIBentrySTDinterwordspacing}{\spaceskip=0pt\relax}
\providecommand{\BIBentryALTinterwordstretchfactor}{4}
\providecommand{\BIBentryALTinterwordspacing}{\spaceskip=\fontdimen2\font plus
\BIBentryALTinterwordstretchfactor\fontdimen3\font minus \fontdimen4\font\relax}
\providecommand{\BIBforeignlanguage}[2]{{%
\expandafter\ifx\csname l@#1\endcsname\relax
\typeout{** WARNING: IEEEtran.bst: No hyphenation pattern has been}%
\typeout{** loaded for the language `#1'. Using the pattern for}%
\typeout{** the default language instead.}%
\else
\language=\csname l@#1\endcsname
\fi
#2}}
\providecommand{\BIBdecl}{\relax}
\BIBdecl

\bibitem{Yao_Jafar_Sum_MAC}
Y.~Yao and S.~A. Jafar, ``The capacity of classical summation over a quantum {MAC} with arbitrarily distributed inputs and entanglements,'' \emph{IEEE Transactions on Information Theory}, vol.~70, no.~9, pp. 6350--6370, 2024.

\bibitem{Schrodinger1935}
E.~Schrodinger, ``Discussion of probability relations between separated systems,'' \emph{Proceedings of the Cambridge Philosophical Society}, p. 31: 555–563; 32 (1936): 446–451, 1936.

\bibitem{Korner_Marton_sum}
J.~Korner and K.~Marton, ``How to encode the modulo-two sum of binary sources (corresp.),'' \emph{IEEE Transactions on Information Theory}, vol.~25, no.~2, pp. 219--221, 1979.

\bibitem{Rai_Dey}
B.~K. Rai and B.~K. Dey, ``On network coding for sum-networks,'' \emph{IEEE Transactions on Information Theory}, vol.~58, no.~1, pp. 50--63, 2012.

\bibitem{Ramamoorthy_Langberg}
A.~Ramamoorthy and M.~Langberg, ``Communicating the sum of sources over a network,'' \emph{IEEE Journal on Selected Areas in Communications}, vol.~31, no.~4, pp. 655--665, 2013.

\bibitem{Appuswamy1}
R.~Appuswamy, M.~Franceschetti, N.~Karamchandani, and K.~Zeger, ``Network coding for computing: Cut-set bounds,'' \emph{IEEE Transactions on Information Theory}, vol.~57, no.~2, pp. 1015--1030, Feb. 2011.

\bibitem{Kowshik_Kumar}
H.~Kowshik and P.~R. Kumar, ``Optimal function computation in directed and undirected graphs,'' \emph{IEEE Transactions on Information Theory}, vol.~58, no.~6, pp. 3407--3418, 2012.

\bibitem{NetworkFC}
X.~Guang, Y.~Bai, and R.~W. Yeung, ``Secure network function computation for linear functions -- part i: Source security,'' 2022.

\bibitem{Guang_Yeung_Yang_Li}
X.~Guang, R.~W. Yeung, S.~Yang, and C.~Li, ``Improved upper bound on the network function computing capacity,'' \emph{IEEE Transactions on Information Theory}, vol.~65, no.~6, pp. 3790--3811, Jun. 2019.

\bibitem{OTAC}
A.~Şahin and R.~Yang, ``A survey on over-the-air computation,'' \emph{IEEE Communications Surveys \& Tutorials}, vol.~25, no.~3, pp. 1877--1908, 2023.

\bibitem{Lloyd}
\BIBentryALTinterwordspacing
V.~Giovannetti, S.~Lloyd, and L.~Maccone, ``Quantum-enhanced measurements: Beating the standard quantum limit,'' \emph{Science}, vol. 306, no. 5700, pp. 1330--1336, 2004. [Online]. Available: \url{https://www.science.org/doi/abs/10.1126/science.1104149}
\BIBentrySTDinterwordspacing

\bibitem{GlobalMetrology}
\BIBentryALTinterwordspacing
T.~J. Proctor, P.~A. Knott, and J.~A. Dunningham, ``Multiparameter estimation in networked quantum sensors,'' \emph{Phys. Rev. Lett.}, vol. 120, p. 080501, Feb 2018. [Online]. Available: \url{https://link.aps.org/doi/10.1103/PhysRevLett.120.080501}
\BIBentrySTDinterwordspacing

\bibitem{DistributedQS}
\BIBentryALTinterwordspacing
Y.~Cao and X.~Wu, ``Distributed quantum sensing network with geographically constrained measurement strategies,'' in \emph{Proc. 2023 IEEE International Conference on Acoustics, Speech and Signal Processing (ICASSP)}, 2023, pp. 1--5. [Online]. Available: \url{https://api.semanticscholar.org/CorpusID:258532409}
\BIBentrySTDinterwordspacing

\bibitem{Young}
\BIBentryALTinterwordspacing
Z.~Ren, W.~Li, A.~Smerzi, and M.~Gessner, ``Metrological detection of multipartite entanglement from young diagrams,'' \emph{Phys. Rev. Lett.}, vol. 126, p. 080502, Feb 2021. [Online]. Available: \url{https://link.aps.org/doi/10.1103/PhysRevLett.126.080502}
\BIBentrySTDinterwordspacing

\bibitem{buhrman1998quantum}
H.~Buhrman, R.~Cleve, and A.~Wigderson, ``Quantum vs. classical communication and computation,'' in \emph{Proceedings of the thirtieth annual ACM symposium on Theory of computing}, 1998, pp. 63--68.

\bibitem{QSM}
D.~Gavinsky, ``Quantum versus classical simultaneity in communication complexity,'' \emph{IEEE Transactions on Information Theory}, vol.~65, no.~10, pp. 6466--6483, 2019.

\bibitem{PSQM}
A.~Kawachi and H.~Nishimura, ``Communication complexity of private simultaneous quantum messages protocols,'' \emph{IACR Cryptology ePrint https://eprint.iacr.org/2021/636.pdf}, May 2021.

\bibitem{song_multiple_server_PIR}
S.~Song and M.~Hayashi, ``Capacity of quantum private information retrieval with multiple servers,'' \emph{IEEE Transactions on Information Theory}, vol.~67, no.~1, pp. 452--463, 2020.

\bibitem{song_colluding_PIR}
------, ``Capacity of quantum private information retrieval with colluding servers,'' \emph{IEEE Transactions on Information Theory}, vol.~67, no.~8, pp. 5491--5508, 2021.

\bibitem{QMDSTPIR}
M.~Allaix, S.~Song, L.~Holzbaur, T.~Pllaha, M.~Hayashi, and C.~Hollanti, ``On the capacity of quantum private information retrieval from {MDS}-coded and colluding servers,'' \emph{IEEE Journal on Selected Areas in Communications}, vol.~40, no.~3, pp. 885--898, 2022.

\bibitem{hayashi2021computation}
M.~Hayashi and {\'A}.~V{\'a}zquez-Castro, ``Computation-aided classical-quantum multiple access to boost network communication speeds,'' \emph{Physical Review Applied}, vol.~16, no.~5, p. 054021, 2021.

\bibitem{aytekin2023quantum}
\BIBentryALTinterwordspacing
A.~Aytekin, M.~Nomeir, S.~Vithana, and S.~Ulukus, ``Quantum symmetric private information retrieval with secure storage and eavesdroppers,'' 2023. [Online]. Available: \url{https://arxiv.org/abs/2308.10883}
\BIBentrySTDinterwordspacing

\bibitem{yard2008capacity}
J.~Yard, P.~Hayden, and I.~Devetak, ``Capacity theorems for quantum multiple-access channels: Classical-quantum and quantum-quantum capacity regions,'' \emph{IEEE Transactions on Information Theory}, vol.~54, no.~7, pp. 3091--3113, 2008.

\bibitem{hsieh2008entanglement}
M.-H. Hsieh, I.~Devetak, and A.~Winter, ``Entanglement-assisted capacity of quantum multiple-access channels,'' \emph{IEEE Transactions on Information Theory}, vol.~54, no.~7, pp. 3078--3090, 2008.

\bibitem{Jafar_FnT}
S.~Jafar, ``Interference alignment: A new look at signal dimensions in a communication network,'' in \emph{Foundations and Trends in Communication and Information Theory}, 2011, pp. 1--136.

\bibitem{Ahlswede_network_information_flow}
R.~Ahlswede, N.~Cai, S.-Y. Li, and R.~Yeung, ``Network information flow,'' \emph{IEEE Transactions on Information Theory}, vol.~46, no.~4, pp. 1204--1216, 2000.

\bibitem{wei2023robust}
H.~Wei, M.~Xu, and G.~Ge, ``Robust network function computation,'' \emph{IEEE Transactions on Information Theory}, vol.~69, no.~11, pp. 7070--7081, 2023.

\bibitem{werner2001all}
R.~F. Werner, ``All teleportation and dense coding schemes,'' \emph{Journal of Physics A: Mathematical and General}, vol.~34, no.~35, p. 7081, 2001.

\bibitem{Allaix_N_sum_box}
M.~Allaix, Y.~Lu, Y.~Yao, T.~Pllaha, C.~Hollanti, and S.~Jafar, ``N-sum box: An abstraction for linear computation over many-to-one quantum networks,'' \emph{IEEE Transactions on Information Theory}, vol.~71, no.~2, pp. 1121--1139, 2025.

\bibitem{grassl2022entropic}
M.~Grassl, F.~Huber, and A.~Winter, ``Entropic proofs of singleton bounds for quantum error-correcting codes,'' \emph{IEEE Transactions on Information Theory}, vol.~68, no.~6, pp. 3942--3950, 2022.

\bibitem{mamindlapally2023singleton}
M.~Mamindlapally and A.~Winter, ``Singleton bounds for entanglement-assisted classical and quantum error correcting codes,'' \emph{IEEE Transactions on Information Theory}, vol.~69, no.~9, pp. 5857--5868, 2023.

\bibitem{Wilde_2017}
M.~M. Wilde, \emph{Quantum Information Theory}, 2nd~ed.\hskip 1em plus 0.5em minus 0.4em\relax Cambridge University Press, 2017.

\bibitem{Peres_QIT}
\BIBentryALTinterwordspacing
A.~Peres and D.~R. Terno, ``Quantum information and relativity theory,'' \emph{Rev. Mod. Phys.}, vol.~76, pp. 93--123, Jan 2004. [Online]. Available: \url{https://link.aps.org/doi/10.1103/RevModPhys.76.93}
\BIBentrySTDinterwordspacing

\bibitem{wiki:No-communication_theorem}
Wikipedia, ``{No-communication theorem} --- {W}ikipedia{,} the free encyclopedia,'' \url{http://en.wikipedia.org/w/index.php?title=No-communication\%20theorem&oldid=1127027017}, 2024, [Online; accessed 11-May-2024].

\bibitem{holevo1973bounds}
A.~S. Holevo, ``Bounds for the quantity of information transmitted by a quantum communication channel,'' \emph{Problemy Peredachi Informatsii}, vol.~9, no.~3, pp. 3--11, 1973.

\bibitem{helwig2012absolute}
W.~Helwig, W.~Cui, J.~I. Latorre, A.~Riera, and H.-K. Lo, ``Absolute maximal entanglement and quantum secret sharing,'' \emph{Physical Review A}, vol.~86, no.~5, p. 052335, 2012.

\bibitem{huber2013structure}
M.~Huber and J.~I. De~Vicente, ``Structure of multidimensional entanglement in multipartite systems,'' \emph{Physical review letters}, vol. 110, no.~3, p. 030501, 2013.

\bibitem{goyeneche2015absolutely}
D.~Goyeneche, D.~Alsina, J.~I. Latorre, A.~Riera, and K.~{\.Z}yczkowski, ``Absolutely maximally entangled states, combinatorial designs, and multiunitary matrices,'' \emph{Physical Review A}, vol.~92, no.~3, p. 032316, 2015.

\bibitem{linden2005new}
N.~Linden and A.~Winter, ``A new inequality for the von neumann entropy,'' \emph{Communications in mathematical physics}, vol. 259, pp. 129--138, 2005.

\bibitem{prabhu2013exclusion}
R.~Prabhu, A.~K. Pati, A.~Sen, U.~Sen \emph{et~al.}, ``Exclusion principle for quantum dense coding,'' \emph{Physical Review A}, vol.~87, no.~5, p. 052319, 2013.

\bibitem{Superdense}
\BIBentryALTinterwordspacing
C.~H. Bennett and S.~J. Wiesner, ``Communication via one- and two-particle operators on {E}instein-{P}odolsky-{R}osen states,'' \emph{Phys. Rev. Lett.}, vol.~69, pp. 2881--2884, Nov 1992. [Online]. Available: \url{https://link.aps.org/doi/10.1103/PhysRevLett.69.2881}
\BIBentrySTDinterwordspacing

\bibitem{araki1970entropy}
H.~Araki and E.~H. Lieb, ``Entropy inequalities,'' \emph{Communications in Mathematical Physics}, vol.~18, no.~2, pp. 160--170, 1970.

\bibitem{pippenger2003inequalities}
N.~Pippenger, ``The inequalities of quantum information theory,'' \emph{IEEE Transactions on Information Theory}, vol.~49, no.~4, pp. 773--789, 2003.

\bibitem{QECtutorial}
\BIBentryALTinterwordspacing
B.~M. Terhal, ``Quantum error correction for quantum memories,'' \emph{Rev. Mod. Phys.}, vol.~87, pp. 307--346, Apr 2015. [Online]. Available: \url{https://link.aps.org/doi/10.1103/RevModPhys.87.307}
\BIBentrySTDinterwordspacing

\bibitem{Ana_Sym}
\BIBentryALTinterwordspacing
A.~C. Silva, \emph{Lectures on Symplectic Geometry}, ser. Lecture Notes in Mathematics.\hskip 1em plus 0.5em minus 0.4em\relax Springer Berlin, Heidelberg, October 2004. [Online]. Available: \url{https://books.google.com/books?id=1VyoPwAACAAJ}
\BIBentrySTDinterwordspacing

\end{thebibliography}
\end{document}